\PassOptionsToPackage{table}{xcolor}
\documentclass[letterpaper,twocolumn,10pt]{article}
\usepackage{usenix2019_v3}

\usepackage{tikz}
\usepackage{amsmath}

\usepackage{filecontents}
\usepackage{enumerate}
\usepackage{mathrsfs}
\usepackage[noend]{algpseudocode}
\usepackage{algorithm}
\usepackage{algorithmicx}
\usepackage[T1]{fontenc}
\usepackage{footmisc}
\usepackage{graphicx}
\usepackage{wrapfig}
\usepackage{enumerate}
\usepackage{hyperref}
\usepackage{amsthm}
\usepackage{subcaption}
\usepackage{multirow}
\usepackage{tablefootnote}
\usepackage{threeparttable}
\usepackage{dblfloatfix}
\usepackage[colorinlistoftodos]{todonotes}
\usepackage{amssymb}
\usepackage{pifont}
\newcommand{\xmark}{\ding{55}}%
\usetikzlibrary{arrows}
\usetikzlibrary{positioning}
\usetikzlibrary{fadings,shapes.arrows,shadows}   
\usetikzlibrary{shapes}
\usepackage{theoremref,amsthm}
\usepackage{blindtext}

\hypersetup{                    
  colorlinks,
  linkcolor={green!80!black},
  citecolor={red!70!black},
  urlcolor={blue!70!black}
}

\definecolor{Red}{HTML}{ff0000}
\definecolor{OrangeRed}{HTML}{FF4500}
\definecolor{FireBrick}{HTML}{B22222}

\tikzfading[name=arrowfading, top color=transparent!0, bottom color=transparent!95]
\tikzset{arrowfill/.style={top color=OrangeRed!20, bottom color=Red, general shadow={fill=black, shadow yshift=-0.8ex, path fading=arrowfading}}}
\tikzset{arrowstyle/.style={draw=FireBrick,arrowfill, single arrow,minimum height=#1, single arrow,
single arrow head extend=.3cm,}}
\tikzstyle{opaque} = [cloud, draw, cloud puffs=10,cloud puff arc=120, aspect=2, inner ysep=0.5em]
\tikzstyle{parse} = [rounded corners=6.75pt, draw, align=center, inner sep=0.175cm]
\tikzstyle{parseinit} = [rounded corners=1.375pt, draw, inner sep=0.05cm]
\tikzstyle{parseop} = [draw, align=center, inner sep=0.175cm]

\newtheorem{lem}{Lemma}

\DeclareMathAlphabet{\mathpzc}{OT1}{pzc}{m}{it}

\tikzstyle{boxtext} = [above, right]
\tikzstyle{data} = [inner sep=0.18cm, draw, align=center, rounded corners]
\tikzstyle{source} = [draw, align=center, dotted]
\tikzstyle{op} = [align=center,inner sep=0.18cm, draw]

\newcommand{\St}{\ensuremath{\mathcal{S}}}
\newcommand{\fu}{\ensuremath{\mathpzc{f}}}
\newcommand{\op}{\ensuremath{\mathpzc{h}}}
\newcommand{\operation}{\ensuremath{\text{op}}}

\newcommand*\Let[2]{\State #1 $\gets$ #2}

\begin{document}

\date{}

\sloppy
\title{\Large \bf Where's Crypto?: Automated Identification and Classification of Proprietary Cryptographic Primitives in Binary Code} 

\author{
{\rm Carlo Meijer}\\
Institute for Computing and Information Sciences\\
Radboud University Nijmegen\\
cmeijer@cs.ru.nl
\and
{\rm Veelasha Moonsamy}\\
Institute for Computing and Information Sciences\\
Radboud University Nijmegen\\
email@veelasha.org
\and
{\rm Jos Wetzels}\\
Midnight Blue Labs\\
j.wetzels@midnightbluelabs.com
} 

\maketitle

\begin{abstract}
The continuing use of proprietary cryptography in embedded systems across many industry verticals, from physical access control systems and telecommunications to machine-to-machine authentication, presents a significant obstacle to black-box security-evaluation efforts. In-depth security analysis requires locating and classifying the algorithm in often very large binary images, thus rendering manual inspection, even when aided by heuristics, time consuming.

In this paper, we present a novel approach to automate the identification and classification of (proprietary) cryptographic primitives within binary code. Our approach is based on Data Flow Graph (DFG) isomorphism, previously proposed by Lestringant et al.~\cite{lestringant2015automated}. Unfortunately, their DFG isomorphism approach is limited to known primitives only, and relies on heuristics for selecting code fragments for analysis. By combining the said approach with symbolic execution, we overcome all limitations of~\cite{lestringant2015automated}, and are able to extend the analysis into the domain of unknown, proprietary cryptographic primitives.
To demonstrate that our proposal is practical, we develop various signatures, each targeted at a distinct class of cryptographic primitives, and present experimental evaluations for each of them on a set of binaries, both publicly available (and thus providing reproducible results), and proprietary ones.
Lastly, we provide a free and open-source implementation of our approach, called \emph{Where's Crypto?}, in the form of a plug-in for the popular IDA disassembler.

\end{abstract}

\section{Introduction}

%

Despite the widely-held academic consensus that cryptography should be publicly documented\cite{gutmann2003,verdult2015security,ker_1883_militaire}, the use of proprietary cryptography has persisted across many industry verticals ranging from physical access control systems~\cite{anderson2006,wouters2019,verdult2015security,strobel2013fuming,weiner2013security,verstegen2018press} and telecommunications~\cite{driessen2012don,nohl2010cryptanalysis,etsi300} to machine-to-machine authentication~\cite{verdult2015security,bokslagassessment}.


This situation presents a significant obstacle to security-evaluation efforts part of certification, compliance, secure procurement or individual research since it requires resorting to highly labor-intensive reverse-engineering in order to determine the presence and nature of these algorithms before they can be evaluated. In addition, when a proprietary algorithm gets broken, details might not be published immediately as a result of NDAs or court injunctions~\cite{bbc2013car} leaving other potentially affected parties to repeat such expensive efforts and hampering effective vulnerability management. As such, there is a real need for practical solutions to automatically scan binaries for the presence of as-of-yet unknown cryptographic algorithms.

\vspace{-3mm}
\paragraph{Criteria}\label{sec:criteria}
In order to support the analysis of closed-source embedded systems for the use of proprietary cryptography, a suitable solution should meet the following criteria: (i) identification of as-of-yet unknown cryptographic algorithms falling within relevant taxonomical classes, (ii) efficient support of large, real-world embedded firmware binaries, and (iii) no reliance on full firmware emulation or dynamic instrumentation due to issues around platform heterogeneity and peripheral emulation. As discussed in Section~\ref{sec:related}, there is no prior work meeting 
all of these criteria.


\vspace{-3mm}
\paragraph{Approach}
To meet the above criteria, our approach bases itself on a structural taxonomy of cryptographic primitives. The idea is that, since the vast majority of proprietary cryptography falls within established primitive classes~\cite{verdult2015security}, we can develop structural signatures allowing for the identification of any algorithm within these classes without having to rely on knowledge of the algorithm's particularites. To this end, we utilize a taxonomy based on~\cite{menezes1996handbook,avanzi2016salad,keliher2003linear,manifavas2016survey} and illustrated in Figure~\ref{fig:taxonomy}. Note that this taxonomy is purely instrumental and does not intend to be exhaustive or allow for an exclusive partitioning of algorithms. 

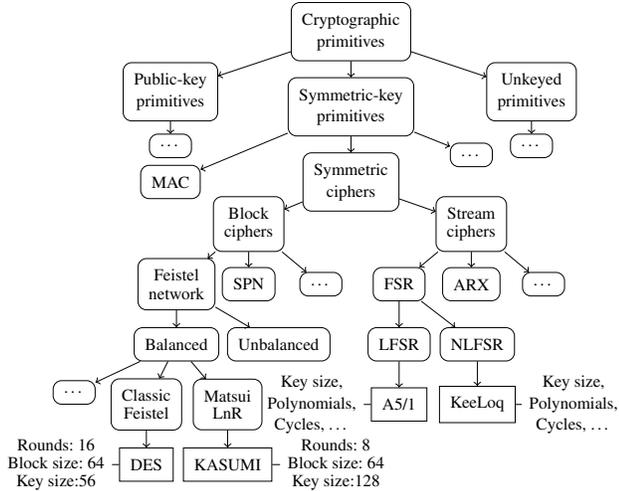
\begin{figure}[!h]
\centering
\vspace*{-0.25cm}
\scalebox{0.8}{
\begin{tikzpicture}[auto]
\node[data] (main) {\footnotesize \shortstack{Cryptographic\\primitives}};
\path (main)+(-3,-1) node[data] (public) {\footnotesize \shortstack{Public-key\\primitives}};
\path (main)+(3,-1) node[data] (unkeyed) {\footnotesize \shortstack{Unkeyed\\primitives}};
\path (main)+(0,-1.25) node[data] (symmetric) {\footnotesize \shortstack{Symmetric-key\\primitives}};
\path[->] (main) edge[] node [] {} (public);
\path[->] (main) edge[] node [] {} (unkeyed);
\path[->] (main) edge[] node [] {} (symmetric);

\path (public)+(0,-0.9) node[data] (pkdots) {\footnotesize \dots};
\path[->] (public) edge[] node [] {} (pkdots);

\path (unkeyed)+(0,-0.9) node[data] (unkdots) {\footnotesize \dots};
\path[->] (unkeyed) edge[] node [] {} (unkdots);

\path (symmetric)+(-3,-1.25) node[data] (mac) {\footnotesize MAC};
\path (symmetric)+(2,-0.8) node[data] (symdots) {\footnotesize \dots};
\path (symmetric)+(0,-1.25) node[data] (ciphers) {\footnotesize \shortstack{Symmetric\\ciphers}};
\path[->] (symmetric) edge[] node [] {} (mac);
\path[->] (symmetric) edge[] node [] {} (ciphers);
\path[->] (symmetric) edge[] node [] {} (symdots);

\path (ciphers)+(-1.7,-0.7) node[data] (block) {\footnotesize \shortstack{Block\\ciphers}};
\path (ciphers)+(2,-0.7) node[data] (stream) {\footnotesize \shortstack{Stream\\ciphers}};
\path[->] (ciphers) edge[] node [] {} (block);
\path[->] (ciphers) edge[] node [] {} (stream);

\path (block)+(-1.2,-1) node[data] (feistel) {\footnotesize \shortstack{Feistel\\network}};
\path (block)+(0,-1) node[data] (spn) {\footnotesize SPN};
\path (block)+(1.2,-1) node[data] (blockdots) {\footnotesize \dots};
\path[->] (block) edge[] node [] {} (feistel);
\path[->] (block) edge[] node [] {} (spn);
\path[->] (block) edge[] node [] {} (blockdots);

\path (feistel)+(0,-1) node[data] (balanced) {\footnotesize Balanced};
\path (feistel)+(1.7,-1) node[data] (unbalanced) {\footnotesize Unbalanced};
\path[->] (feistel) edge[] node [] {} (balanced);
\path[->] (feistel) edge[] node [] {} (unbalanced);

\path (balanced)+(-0.5,-1) node[data] (classic) {\footnotesize \shortstack{Classic\\Feistel}};
\path (balanced)+(0.85,-1) node[data] (matsui) {\footnotesize \shortstack{Matsui\\LnR}};
\path (balanced)+(-1.7,-0.8) node[data] (balanceddots) {\footnotesize \dots};
\path[->] (balanced) edge[] node [] {} (classic);
\path[->] (balanced) edge[] node [] {} (matsui);
\path[->] (balanced) edge[] node [] {} (balanceddots);

\path (stream)+(-1.2,-1) node[data] (fsr) {\footnotesize FSR};
\path (stream)+(0,-1) node[data] (arx) {\footnotesize ARX};
\path (stream)+(1.2,-1) node[data] (streamdots) {\footnotesize \dots};
\path[->] (stream) edge[] node [] {} (fsr);
\path[->] (stream) edge[] node [] {} (arx);
\path[->] (stream) edge[] node [] {} (streamdots);

\path (fsr)+(0,-1) node[data] (lfsr) {\footnotesize LFSR};
\path (fsr)+(1.3,-1) node[data] (nlfsr) {\footnotesize NLFSR};
\path[->] (fsr) edge[] node [] {} (lfsr);
\path[->] (fsr) edge[] node [] {} (nlfsr);

\path (classic)+(0,-1) node[op] (des) {\footnotesize DES};
\path[->] (classic) edge[] node [] {} (des);

\path (matsui)+(0,-1) node[op] (kasumi) {\footnotesize KASUMI};
\path[->] (matsui) edge[] node [] {} (kasumi);

\path (lfsr)+(0,-1) node[op] (a51) {\footnotesize A5/1};
\path[->] (lfsr) edge[] node [] {} (a51);

\path (nlfsr)+(0,-1) node[op] (keeloq) {\footnotesize KeeLoq};
\path[->] (nlfsr) edge[] node [] {} (keeloq);

\path (des)+(-1.5,0) node[] (desspec) {\footnotesize \shortstack{Rounds: 16\\Block size: 64\\Key size:56}};
\path[-] (desspec) edge[] node[] {} (des);

\path (kasumi)+(1.8,0) node[] (kasumispec) {\footnotesize \shortstack{Rounds: 8\\Block size: 64\\Key size:128}};
\path[-] (kasumispec) edge[] node[] {} (kasumi);

\path (a51)+(-1.45,0) node[] (a51spec) {\footnotesize \shortstack{Key size,\\Polynomials,\\Cycles, \dots}};
\path[-] (a51spec) edge[] node[] {} (a51);

\path (keeloq)+(1.6,0) node[] (keeloqspec) {\footnotesize \shortstack{Key size,\\Polynomials,\\Cycles, \dots}};
\path[-] (keeloqspec) edge[] node[] {} (keeloq);

\end{tikzpicture}
}
 \caption{Taxonomical tree of algorithm classes}
 \label{fig:taxonomy}
\end{figure}

Our approach is built on two fundamentals: Data Flow Graph (DFG) isomorphism and symbolic execution. As described in Section~\ref{sec:overview}, the limitations of prior work on DFG isomorphism~\cite{lestringant2015automated}
are overcome through augmentation with symbolic execution which allows us to specify structural signatures for taxonomic classes of cryptographic primitives and analyze binary code for matches. The focus of this paper is on symmetric and unkeyed primitives.

\vspace{-3mm}
\paragraph{Contribution}
Our contribution is threefold. First, our novel approach combines subgraph isomorphism with symbolic execution, solving the open problem of fragment selection and eliminating the need for heuristics and thus, overcoming the limitations of prior work which rendered it unsuited to identifying unknown ciphers. To the best of our knowledge, as discussed in Section~\ref{sec:related}, there is currently no prior work in either industry or academia that addresses the problem of identifying unknown cryptographic algorithms. Second, we propose a new domain-specific language (DSL) for defining the structural properties of cryptographic primitives, along with several examples. Finally, a free and open-source proof-of-concept (PoC) implementation, \emph{Where's Crypto?}, is made available\footnote{\url{https://github.com/wheres-crypto/wheres-crypto}} and evaluated in terms of analysis time and accuracy against relevant real-world binaries.

\section{Scope and limitations}\label{sec:scope}

\paragraph{Normalization and optimization}
A single function can be represented as many different combinations of assembly instructions depending on architecture and compiler particularities. Attempting to construct a 1--to--1 mapping between semantic equivalence classes and DFGs is beyond the scope of this work. When our normalization maps two expressions to the same DFG node, they are considered to be semantically equivalent. While the inverse is not necessarily true, our approach can operate as if this were the case since, for a compiler to take advantage of semantic equivalences, it must be consistently aware of them. Therefore, we can leverage this fact to recognize compiler-generated equivalences.

\vspace{-3mm}
\paragraph{Implicit flows}\label{sec:implicit} Data dependencies may also arise due to control-dependent assignments. For  example, given two boolean variables $a$ and $b$, statements $a \leftarrow b$ and $\texttt{if~} a \texttt{~then~} b \leftarrow \textbf{~true}; \texttt{~else~} b \leftarrow \textbf{~false}$ are semantically equivalent. In the former, $b$ directly flows to $a$, and therefore the dependency is apparent in its corresponding DFG, whereas in the latter, the dependency information is lost. Since data-dependent branches increase side-channel susceptibility, developers should refrain from using them for cryptographic primitives. Therefore, we believe it is justified to declare implicit flows out of scope. Note that implicit flows is a concept different from data-dependent branches. Support for the latter is achieved by means of symbolic execution (Section~\ref{sec:symbolic}).


\vspace{-3mm}
\paragraph{Function entry points} 
Our PoC implementation relies on IDA's recognition of function entry points as input to our algorithm. As such, inaccuracies in IDA's function recognition will reduce our coverage. However, this is not an inherent limitation of our approach but merely of the implementation.

\vspace{-3mm}
\paragraph{Code obfuscation} Since code obfuscation presents an inherent challenge to any binary-analysis approach, our approach assumes that the input it operates on is not obfuscated and delegates this de-obfuscation to a manual and/or automated pre-processing step. Automated binary deobfuscation is a well-established research field of its own which consists of a wide variety of static, dynamic, symbolic and concolic approaches \cite{yadegari2016automatic,david2017formal,salwan2018symbolic,xu2018vmhunt} drawing upon synthesis \cite{blazytko2017syntia,biondi2017effectiveness}, optimization \cite{garba2019saturn}, semantic equivalence \cite{tofighi2018dose} and machine learning \cite{tofighi2019defeating} based techniques in order to make obfuscated binaries amenable to analysis.

\vspace{-3mm}
\paragraph{Taxonomical constraints}
In our PoC evaluation and the examples of our DSL, we have limited our discussion to a subset of the taxonomy of cryptographic primitives. This is not an inherent limitation of our approach, but merely of our PoC and its evaluation. Our approach is essentially agnostic with respect to the employed taxonomy, which can be extended as users see fit, and only assumes that the algorithm the analyst is looking for is within one of its classes. Given that the vast majority of proprietary cryptography falls within a specific subset of established primitive classes~\cite{verdult2015security}, namely stream- and block ciphers and hash functions, we do not consider this a practical issue.

\vspace{-3mm}
\paragraph{False positives}
Certain primitive classes are a subset of others and some instances fit the definition of several ones. As such, their matches are prone to false positives. Examples of such are discussed in Section~\ref{sec:falsepositives}. We do not consider this a serious practical problem as our solution is intended to assist a human analyst who will be easily capable of pruning a limited number of false positives compared to the burden of unassisted analysis required by the status quo.

Furthermore, certain primitive classes are essentially underdefined. That is to say, their definition is so broad that characteristic properties are not distinctive enough for a meaningful identification. For example, the defining property of stream ciphers is two data streams being XOR-ed together. Obviously, identifying instances of XOR results in an overwhelming number of false positives. In case a signature for such a generic class is desired, an alternative approach is to craft signatures for every subclass contained within it.

\vspace{-3mm}
\paragraph{Path oracle policy}\label{sec:pathoraclelimitations} 
The path oracle policy discussed in Section~\ref{sec:pathoracle} is chosen such that the resulting graph represents $n$ iterations of an algorithm. While this typically satisfies our goals, there are a few exceptions to this rule. First, compilers sometimes ensure loop-guard evaluation during both entry and exit, resulting in a DFG representing $n+1$ iterations. Second, cryptographic primitives with a constant iteration length are beyond the control of the path oracle. Finally, loop unrolling will result in a DFG representing $kn$ iterations, where $k$ denotes the number of compiler-grouped iterations. In order to overcome this limitation, we suggest taking the possibility of iteration count deviating from $n$ into account during signature construction as described in Section~\ref{sec:signatures}, for example by defining a minimum rather than an exact match.

\section{Prior work}\label{sec:related}
Prior work by academia and industry into the identification of cryptographic algorithms in binary code can be divided into (combinations of) the following approaches:

\vspace{-3mm}
\paragraph{Dedicated functionality identification} The most naive and straight-forward approach consists of identifying dedicated cryptographic functionality in the form of OS APIs (e.g. Windows CryptoAPI/CNG)~\cite{matenaar2012cis}, library imports or dedicated instructions (e.g. AES-NI). This approach is inherently incapable of detecting unknown algorithms. 

\vspace{-3mm}
\paragraph{Data signatures} The most common approach employed in practice~\cite{auriemma2013signsrch,plohmann2012simplifire,guilfanovfindcrypt,snaker2015kanal,loki2008snd,x3chun2004crypto,paradox2009hcd,levin2013draft} consists of identifying cryptographic algorithms on the basis of constants (e.g. IVs, Nothing-Up-My-Sleeve Numbers, padding) and lookup tables (e.g. S-Boxes, P-Boxes). The approach is unsuitable for detecting unknown algorithms. Moreover, the same applies for known algorithms that do not rely on fixed data, or those that do, but, for example, use dynamically generated S-Boxes, rather than embedded ones.



\vspace{-3mm}
\paragraph{Code heuristics} Another series of approaches rely on code heuristics, which are applied either statically or dynamically, like mnemonic-constant tuples~\cite{grobert2011,lagadec2014balbuzard}, which take into account word sizes, endianness, and multiplicative and additive inverses but otherwise suffer from the same drawbacks as data signatures.

A second heuristic relies on the observation that symmetric cryptographic routines tend to consist of a high ratio of bitwise arithmetic instructions~\cite{caballero2009dispatcher,plohmann2012simplifire,lagadec2014balbuzard,grobert2011,matenaar2012cis} and attempt to classify functions based on a threshold. The drawback of this approach is that it lacks granular taxonomical identification capabilities as well as being highly prone to false positives, especially on embedded systems where heavy bitwise arithmetic is typically present as part of memory-mapped register operations required for peripheral interaction.

\vspace{-3mm}
\paragraph{Deep learning} Hill et al.~\cite{hill2017deep} propose a Dynamic Convolutional Neural Network based approach which, however, is unsuited for our purposes due to its reliance on dynamic binary instrumentation and its inherent inability to classify unknown algorithms.

\vspace{-3mm}
\paragraph{Data flow analysis}
One set of approaches to data flow analysis relies on the static relation between functions and their inputs and outputs~\cite{grobert2011,newsome2005dynamic,calvet2012aligot,matenaar2012cis}. One plausible approach is to perform taint analysis and evaluate function I/O entropy changes, which relies on emulation and as such is unsuitable as per our criteria in Section~\ref{sec:criteria}. Another approach is to compare emulated or symbolically executed function I/O to a collection of reference implementations or test vectors, which is inherently incapable of detecting unknown algorithms.

Another approach~\cite{xu2017cryptographic} utilizes dynamic instrumentation and symbolic execution to translate candidate cryptographic algorithms into boolean formulas for subsequent comparison to reference implementations using guided fuzzing. However, its reliance on dynamic instrumentation and inherent inability to recognize unknown algorithms render the approach unsuitable for our purposes.

Finally, there is the DFG isomorphism approach as proposed by~\cite{lestringant2015automated} which produces DFGs from a given binary and compares it against graphs of known cryptographic algorithms through the use of Ullmann's subgraph isomorphism algorithm\cite{ullmann1976}.
A DFG is a Directed Acyclic Graph (DAG) representing the flow of data within a sequence of arithmetic/logic operations. A vertex represents either an operation, or an input variable. The presence of an edge between vertex $v_1$ and $v_2$ means that $v_1$ (or the result of operation $v_1$) is an input to operation $v_2$.
Due to the nature of DFGs, code flow information cannot be expressed. As such, the contributions of~\cite{lestringant2015automated} are limited to linear sequences of instructions. Moreover, the authors argue that since cryptographic implementations ought to avoid data-dependent branching due to side-channel susceptibility, one can assume all cryptographic code is free from data-dependent conditional instructions. This latter generalization introduces several limitations.

First, no straightforward strategy for selecting code fragments is proposed. Performing the analysis on a per-function basis is complicated by the fact that cryptographic implementations are commonly surrounded by some basic control logic, such as checks on input parameters. As a result, analysis can neither be applied to entire functions nor across function boundaries through inlining and hence the authors propose a limited set of selection heuristics constraining the work.

Second, the approach performs well when identifying known algorithms since one can take advantage of algorithm-unique characteristics, but this does not hold when attempting to identify unknown algorithms. Furthermore, a common pattern is that the class of a cryptographic primitive often only becomes apparent once the analysis incorporates conditional instructions. We clarify this point using the following toy examples.

Suppose that we would like to identify a proprietary stream cipher $\sigma$. A typical implementation contains a key-stream generator, generating pseudo-random bytes in a loop. Inevitably, this loop contains a conditional instruction causing the program to either re-enter or exit the loop, depending on the length parameter. As there is no support for conditional instructions depending on non-constant values, DFG $G$, generated from $\sigma$ will, at most, represent a single iteration, covering a single unit of input length (bytes or otherwise). In this typical example, clearly, a stream cipher pattern will not become apparent in $G$. The example can be generalized to any pattern that becomes apparent only after several iterations, where no additional properties of the target primitive are known.

Similarly, suppose that we would like to identify a proprietary hash function $\theta$, based on a Merkle-Damg{\aa}rd construction. $\theta$ invokes compression function $F$, which processes blocks of fixed input length. The Merkle-Damg{\aa}rd construction is then used to allow variable input lengths. As such, in order to generate a DFG wherein the construction is apparent, we need it to incorporate several iterations, and perform inlining of $F$. The former is problematic (as per the stream cipher example), and so is the latter in case $F$ performs some kind of input validation, for e.g. checking for \textsc{null} pointers.

\section{Solution overview}\label{sec:overview}
Cryptographic primitives are essentially a set of arithmetic and logical operations representing an input/output relation. This structural relationship between operations and data can be expressed as a DFG. Since all particular algorithms will be structurally similar to the general primitive defining their taxonomical class, the problem of identifying an unknown algorithm assumed to belong to a well-defined taxonomical class can be formulated as a DFG subgraph isomorphism problem. However, due to slight differences in implementation and compiler peculiarities, DFG representations of semantically identical algorithms may differ and such representations require normalization before they can be subjected to isomorphism analysis.
Lestringant et al.~\cite{lestringant2015automated} demonstrated that, by repeatedly applying a set of rewrite rules to the DFG, a normalized version is obtained, wherein many of these variations are removed. Although no guarantee can be given that equivalent semantics will always map to the same DFG, the result is `good enough' to serve as a data structure for the purpose.

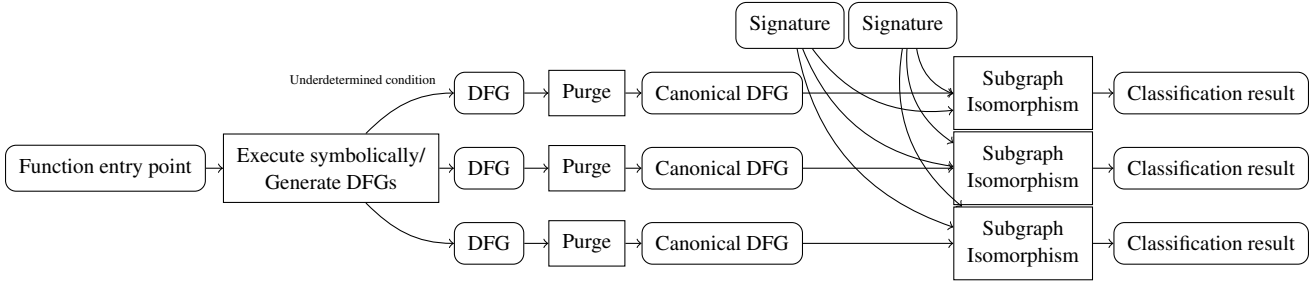
\begin{figure*}[tb]
\centering
   \begin{tikzpicture}[auto]
    \node[data] (entry) {\footnotesize Function entry point};
    \path (entry)+(3,0) node[op] (symex) {\footnotesize \shortstack{Execute symbolically/\\Generate DFGs}};
    \path (entry)+(5.1,1) node[data] (dfg0) {\footnotesize DFG};
    \path (entry)+(5.1,0) node[data] (dfg1) {\footnotesize DFG};
    \path (entry)+(5.1,-1) node[data] (dfg2) {\footnotesize DFG};
    \path (dfg0)+(1.3,0) node[op] (clean0) {\footnotesize Purge};
    \path (dfg1)+(1.3,0) node[op] (clean1) {\footnotesize Purge};
    \path (dfg2)+(1.3,0) node[op] (clean2) {\footnotesize Purge};
    \path (clean0)+(1.8,0) node[data] (cdfg0) {\footnotesize Canonical DFG};
    \path (clean1)+(1.8,0) node[data] (cdfg1) {\footnotesize Canonical DFG};
    \path (clean2)+(1.8,0) node[data] (cdfg2) {\footnotesize Canonical DFG};

    \path (cdfg0)+(0.9,0.9) node[data] (sign0) {\footnotesize \shortstack{Signature}};
    \path (cdfg0)+(2.4,0.9) node[data] (sign1) {\footnotesize \shortstack{Signature}};
    
    \path (cdfg0)+(4,0) node[op] (ism0) {\footnotesize \shortstack{Subgraph\\Isomorphism}};
    \path (cdfg1)+(4,0) node[op] (ism1) {\footnotesize \shortstack{Subgraph\\Isomorphism}};
    \path (cdfg2)+(4,0) node[op] (ism2) {\footnotesize \shortstack{Subgraph\\Isomorphism}};
    \path (ism0)+(2.5,0) node[data] (result0) {\footnotesize Classification result};
    \path (ism1)+(2.5,0) node[data] (result1) {\footnotesize Classification result};
    \path (ism2)+(2.5,0) node[data] (result2) {\footnotesize Classification result};

    \path[->] (entry) edge[] node [] {} (symex);
    \path[->] (symex) edge[out=45,in=180,pos=0.92] node [] {\tiny Underdetermined condition} (dfg0);
    \path[->] (symex) edge[] node [] {} (dfg1);
    \path[->] (symex) edge[out=-45,in=180] node [] {} (dfg2);
    \path[->] (dfg0) edge[] node [] {} (clean0);
    \path[->] (dfg1) edge[] node [] {} (clean1);
    \path[->] (dfg2) edge[] node [] {} (clean2);
    \path[->] (clean0) edge[] node [] {} (cdfg0);
    \path[->] (clean1) edge[] node [] {} (cdfg1);
    \path[->] (clean2) edge[] node [] {} (cdfg2);
    \path[->] (cdfg0) edge[] node [] {} (ism0);
    \path[->] (cdfg1) edge[] node [] {} (ism1);
    \path[->] (cdfg2) edge[] node [] {} (ism2);
    \path[->] (sign0) edge[bend right=30] node [] {} (ism0);
    \path[->] (sign0) edge[bend right=30] node [] {} (ism1);
    \path[->] (sign0) edge[bend right=30] node [] {} (ism2);
    \path[->] (sign1) edge[bend right=30] node [] {} (ism0);
    \path[->] (sign1) edge[bend right=30] node [] {} (ism1);
    \path[->] (sign1) edge[bend right=30] node [] {} (ism2);
    \path[->] (ism0) edge[] node [] {} (result0);
    \path[->] (ism1) edge[] node [] {} (result1);
    \path[->] (ism2) edge[] node [] {} (result2);
  \end{tikzpicture}
  \caption{Diagram of primitive identification process}
  \label{fig:overview}
\vspace{-0.33cm}
\end{figure*}

The identification procedure consists of three stages. A diagram of the procedure is given in Figure~\ref{fig:overview}. First, given the entry point of a function, we start executing it symbolically. A DFG is constructed during the execution, where each instruction adds a set of nodes and edges to the graph. In case a conditional instruction is encountered, the execution path belonging to the condition evaluating to \textbf{true}, \textbf{false}, or both paths are explored. In the latter case, the partially constructed DFG is duplicated and the construction continues independently for both execution paths. Hence, the final result of the DFG construction phase is, in fact, a set of DFGs describing the input/output relation corresponding to the execution path taken. Section~\ref{sec:dfgconstruct} describes the construction phase in detail.

Second, once a DFG is fully constructed, we enter the purging phase. This phase is responsible for removing nodes from the graph that represent neither an output, nor a value used in the computation of any output. As such, the graph is reduced to a form in which it only represents the input/output relation, free from operations introduced due to register spilling and other possible implementation, compiler, and architecture-specific operations that are irrelevant to the function's semantics. Section~\ref{sec:dfgpurge} describes the purging phase in detail.

Last, with the finalized DFG at our disposal, we enter the pattern-matching phase, where we search for subgraphs in the DFG that are isomorphic to the graph signature of a given cryptographic primitive. If such a subgraph is identified, we conclude that the primitive is indeed present in the instructions from which the DFG was generated. We use Ullmann's subgraph isomorphism algorithm for searching the DFG. Section~\ref{sec:dfgsearch} describes the pattern-matching phase in detail.

\section{Data Flow Graph construction}\label{sec:dfgconstruct}
The approach of constructing the DFG from assembly instructions builds upon that of~\cite{lestringant2015automated}. This section summarizes their approach, and indicates where ours departs from it.

Suppose we have a sequence of assembly instructions. We construct its corresponding DFG, $G = (V,E)$, by converting each instruction $i$ into a set of operations $O_i$, which can potentially be empty (e.g., a \texttt{NOP} or branch), or contain multiple operations (e.g., a complex instruction). We distinguish three cases based on input type, as follows:
\vspace{-3mm}
\paragraph{Immediate} We create a vertex representing a constant value in $G$. It is linked by an edge to $O_i$.
 \vspace{-3mm}
\paragraph{Register} In case an instruction takes a register as an input operand, we create an edge between the last value written to that register and $O_i$. In practice, this means we maintain an array containing, for each register, a reference to the vertex in $G$ corresponding to that value.
 \vspace{-3mm}
\paragraph{Memory} For operands that load or store from/to memory, we create \texttt{LOAD} and \texttt{STORE} operations. Both operations take a memory address vertex as input. Like any other vertex, the address can be a constant, or a more complex symbolic expression.

Ideally, we would like all code fragments within a semantic equivalence class to map to the same DFG, and have the end result represent the semantics only, free from architecture and compiler-specific traits.
The approach followed by~\cite{lestringant2015automated} is to take the generated DFG, and repeatedly apply normalization rewrite rules until a fixed-point is reached. This is where our approach deviates from theirs, as we apply normalization as well, but \emph{continuously} during graph construction. This enhances performance, which we argue below in Section~\ref{sec:dfgconstructadvantages}, and allows us to efficiently keep track of the conditions that apply during symbolic execution (Section~\ref{sec:symbolic}).

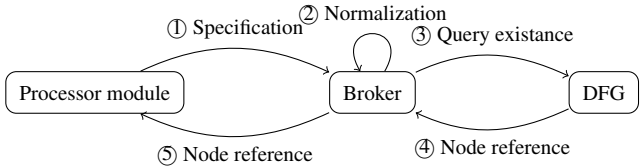
\begin{figure}[htb]
    \begin{tikzpicture}[auto]
    \node[data] (proc) {\footnotesize Processor module};
    \path (proc.east)+(2.5,0) node[data] (broker) {\footnotesize Broker};
    \path (broker.east)+(2.5,0) node[data] (dfg) {\footnotesize DFG};

    \path[->] (proc) edge[in=155,out=25,above] node [] {\footnotesize \textcircled{1} Specification} (broker);
    \path[<-] (proc) edge[in=205,out=335,below] node [] {\footnotesize \textcircled{5} Node reference} (broker);
    \path[->] (broker) edge[in=120,out=60,looseness=6,above] node [] {\footnotesize \textcircled{2} Normalization} (broker);
    \path[->] (broker) edge[in=155,out=25] node[] {\footnotesize \textcircled{3} Query existance} (dfg);
    \path[<-] (broker) edge[in=205,out=335,below] node [] {\footnotesize \textcircled{4} Node reference} (dfg);

  \end{tikzpicture}
\caption{Flow of the graph-node creation process}
  \label{fig:broker}
\end{figure}

A diagram of the graph-node creation process is given in Figure~\ref{fig:broker}. More concretely: there is a \emph{processor module}, written for a specific architecture that translates each instruction into graph nodes. The processor module cannot autonomously create new graph nodes. Instead, it must interact with the \emph{broker}. The broker is responsible for the application of normalization rewrite rules and is processor-architecture agnostic.
The processor module provides a specification of the desired node to the broker, which in turn applies normalization rewrite rules to the specification. As such, the result either matches the specification exactly, or a different one that is semantically equivalent.
After normalization, the broker queries the DFG for whether a node conforming to the normalized specification already exists. If it does, a reference to it is returned, rather than a new node being created.
Consequently, there cannot exist two distinct nodes in a graph conforming to the same specification, or equivalent under normalization. We prove this property in Lemma~\ref{lem:normalized}.

\begin{lem}
\label{lem:normalized}
Let $G = (V, E)$ be a DFG, and $\op$ denote the normalization transform, for which holds: (1) $\op(\op(x)) = \op(x)$ for all $x \in U$ (universe).
Consider arbitrary arithmetic/logical operation $\operation(v_1,v_2)$, where $v_1,v_2 \in V$.

A broker request for $\operation$ preserves the following properties:
(i) For all $v \in V$, $v = \op(v)$, i.e. all nodes in $G$ are normalized.
(ii) For all $v_1, v_2 \in V$, $\op(v_1) = \op(v_2) \implies v_1 = v_2$, i.e. all nodes in $G$ belong to a unique equivalence class under the normalization function.
\end{lem}
\begin{proof}
Assume (i) and (ii) hold for $V$. We define $q = \op(\operation(v_1,v_2))$ and distinguish two cases.

If $q \in V$, then $G$ is not modified and (i) and (ii) are trivially preserved.
If $q \not\in V$, then $V' = V \cup \{q\}$.
By applying (1), we get $\op(q) = q$, and thus (i) holds for $\{q\}$. Since (i) already holds for $V$, (i) also holds for $V'$.
Furthermore, suppose that there exists $p \in V$, for which $\op(p) = \op(q)$. By (i), we get $\op(p) = p$, and hence $p = \op(q)$. By definition, $q = \op(\operation(v_1,v_2))$ and hence $p = \op(\op(\operation(v_1,v_2)))$. By (1), we get $p = \op(\operation(v_1,v_2))$ and thus $p = q$. This contradicts $q \not\in V$, and hence no $p \in V$ exists such that $\op(p) = \op(q)$. Therefore, (ii) holds for $V'$.

Since (i) and (ii) trivially hold for the base case, i.e., an empty graph $G$, where $V = \varnothing$, and the above shows preservation during the step case, the properties hold for any $G$. 
\end{proof}




At this point, we are ready to describe the normalization rewrite rules; they include operation simplification, common-subexpression elimination, and subsequent memory access.

\vspace*{-0.35cm}\paragraph{Operation simplification} Suppose that we encounter an arithmetic/logic operation for which all input parameters are constants. Then, the operation can be replaced by its result.

\vspace*{-0.25cm}\begin{center}
\scalebox{0.8}{
\begin{tikzpicture}[auto]
\node[data] (4) {\footnotesize \texttt{4}};
\path (4)+(0.75,0) node[data] (12) {\footnotesize \texttt{12}};
\path (4)+(0.375, -0.8) node[op,circle,inner sep=0.1cm] (add) {+};

\path[->] (4) edge[] node[] {} (add);
\path[->] (12) edge[] node[] {} (add);

\path (add)+(1.5,0.4) node [] (arrow1) {};
\path (arrow1.center)+(1,0) node[] (arrow2) {};
\path[->,line width=3pt] (arrow1.center) edge [] node[] {} (arrow2.center);

\path (arrow2)+(1,0) node [data] (16) {\footnotesize \texttt{16}};
\end{tikzpicture}
}
\end{center}

\vspace*{-0.35cm}\noindent Likewise, in case an element is the identity element for the operation it serves as an input to, the operation has no effect and can be removed. In case an element is the zero element, the operation can be replaced by zero.

%
%
%
%
%
%
%
%

\vspace*{-0.35cm}\paragraph{Common subexpression elimination} Often within a code fragment, the same value is re-computed several times. This is especially true when the instruction set allows for expressing complex operands, for e.g. supporting offsets and shifts. 
Lemma~\ref{lem:normalized} states that broker requests for nodes belonging to a certain equivalence class all result in references to the same graph node. 
Hence, common-subexpression elimination is already achieved by the design of the node-creation process.

\vspace*{-0.25cm}\begin{center}
\scalebox{0.8}{
\begin{tikzpicture}[auto]
\node[data] (sp1) {\footnotesize SP};
\path (sp1)+(0.375,-0.8) node[op,circle,inner sep=0.1cm] (add1) {+};
\path (sp1)+(0.75,0) node[op,circle,inner sep=0.1cm] (lsh1) {\footnotesize \textsc{\textless{}\textless{}}};
\path (lsh1)+(-0.375,0.8) node[data] (a1) {\footnotesize R2};
\path (a1)+(0.75,0) node[data] (b1) {\footnotesize \texttt{2}};

\path (lsh1)+(1.5,0) node[data] (sp2) {\footnotesize R0};
\path (sp2)+(0.375,-0.8) node[op,circle,inner sep=0.1cm] (add2) {+};
\path (sp2)+(0.75,0) node[op,circle,inner sep=0.1cm] (lsh2) {\footnotesize \textsc{\textless{}\textless{}}};
\path (lsh2)+(-0.375,0.8) node[data] (a2) {\footnotesize R2};
\path (a2)+(0.75,0) node[data] (b2) {\footnotesize \texttt{2}};

\path[->] (sp1) edge[] node[] {} (add1);
\path[->] (lsh1) edge[] node[] {} (add1);
\path[->] (a1) edge[] node[] {} (lsh1);
\path[->] (b1) edge[] node[] {} (lsh1);

\path[->] (sp2) edge[] node[] {} (add2);
\path[->] (lsh2) edge[] node[] {} (add2);
\path[->] (a2) edge[] node[] {} (lsh2);
\path[->] (b2) edge[] node[] {} (lsh2);

\path (lsh2)+(1,0) node [] (arrow1) {};
\path (arrow1.center)+(1,0) node[] (arrow2) {};
\path[->,line width=3pt] (arrow1.center) edge [] node[] {} (arrow2.center);

\path (arrow2)+(0.75,0.2) node[data] (sp3) {\footnotesize SP};
\path (sp3)+(0.375,-0.8) node[op,circle,inner sep=0.1cm] (add3) {+};
\path (sp3)+(0.9,-1.2) node[op,circle,inner sep=0.1cm] (add4) {+};
\path (sp3)+(0.75,0) node[op,circle,inner sep=0.1cm] (lsh3) {\footnotesize \textsc{\textless{}\textless{}}};
\path (lsh3)+(0.65,-0.4) node[data] (sp4) {\footnotesize R0};
\path (lsh3)+(-0.2,0.8) node[data] (a3) {\footnotesize R2};
\path (a3)+(0.75,0) node[data] (b3) {\footnotesize \texttt{2}};

\path[->] (sp3) edge[] node[] {} (add3);
\path[->] (lsh3) edge[] node[] {} (add3);
\path[->] (a3) edge[] node[] {} (lsh3);
\path[->] (b3) edge[] node[] {} (lsh3);
\path[->] (lsh3) edge[] node[] {} (add4);
\path[->] (sp4) edge[] node[] {} (add4);
\end{tikzpicture}
}
\end{center}

\vspace*{-0.35cm}\paragraph{Memory access} Loading and storing of data from/to main memory is a common operation. However, this need not have a relation with semantics, but may be due to register filling and spilling. We attempt to correct for this by substituting each \texttt{LOAD} operation by its result, which is known in case a preceding \texttt{STORE} operation to the same memory address node exists.
It is important to be able to identify the potential equivalence of memory address nodes passed to the \texttt{STORE} and \texttt{LOAD} operation. Like any other expression, memory addresses are represented by graph nodes. Given Lemma~\ref{lem:normalized}, all equivalent address nodes are mapped to a single graph node.
By maintaining a lookup table during graph construction, for e.g., a hash table mapping address nodes to their corresponding stored value, the substitution can be performed in constant time.

\vspace*{-0.25cm}\begin{center}
\scalebox{0.8}{
\begin{tikzpicture}[auto]
\node[data] (r3) {\footnotesize R3};
\path (r3)+(0.75,0) node[op,circle,inner sep=0.1cm] (add1) {+};
\path (add1)+(-0.375,0.8) node[data] (sp1) {\footnotesize SP};
\path (sp1)+(0.75,0) node[data] (8) {\footnotesize \texttt{8}};
\path (r3)+(0.375,-0.8) node[op] (store) {\footnotesize \texttt{STORE}};

\path (add1)+(2,0) node[op,circle,inner sep=0.1cm] (add2) {+};
\path (add2)+(-0.375,0.8) node[data] (sp2) {\footnotesize SP};
\path (sp2)+(0.75,0) node[data] (8p) {\footnotesize 8};
\path (add2)+(0.2,-0.8) node[op] (load) {\footnotesize \texttt{LOAD}};
\path (load)+(-0.5,-0.8) node[op,circle,inner sep=0.05cm] (and) {\footnotesize \textsc{and}};
\path (load)+(-1,0) node[data] (const) {\footnotesize \texttt{0xff}};

\path[->] (r3) edge[] node[] {} (store);
\path[->] (add1) edge[] node[] {} (store);
\path[->] (sp1) edge[] node[] {} (add1);
\path[->] (8) edge[] node[] {} (add1);

\path[->] (add2) edge[] node[] {} (load);
\path[->] (sp2) edge[] node[] {} (add2);
\path[->] (8p) edge[] node[] {} (add2);
\path[->] (load) edge[] node[] {} (and);
\path[->] (const) edge[] node[] {} (and);

\path (add2)+(1,-0.4) node [] (arrow1) {};
\path (arrow1.center)+(1,0) node[] (arrow2) {};
\path[->,line width=3pt] (arrow1.center) edge [] node[] {} (arrow2.center);

\path (arrow2)+(2.05,0.4) node[op,circle,inner sep=0.1cm] (add3) {+};
\path (add3)+(-0.75,0) node[data] (r3p) {\footnotesize R3};
\path (add3)+(-0.6,0.8) node[data] (sp3) {\footnotesize SP};
\path (sp3)+(0.75,0) node[data] (8pp) {\footnotesize \texttt{8}};
\path (r3p)+(0.65,-0.8) node[op] (storep) {\footnotesize \texttt{STORE}};
\path (r3p)+(0.1,-1.6) node[op,circle,inner sep=0.05cm] (andp) {\footnotesize \textsc{and}};
\path (r3p)+(-0.5,-0.8) node[data] (constp) {\footnotesize \texttt{0xff}};

\path[->] (r3p) edge[] node[] {} (storep);
\path[->] (add3) edge[] node[] {} (storep);
\path[->] (sp3) edge[] node[] {} (add3);
\path[->] (8pp) edge[] node[] {} (add3);
\path[->] (r3p) edge[] node[] {} (andp);
\path[->] (constp) edge[] node[] {} (andp);
\end{tikzpicture}
}
\end{center}

\vspace*{-0.35cm}For associative operations, the result does not depend on the order in which they are executed. Therefore we translate nested associative operations into a single operation taking all inputs.

\vspace*{-0.25cm}\begin{center}
\scalebox{0.8}{
\begin{tikzpicture}[auto]
\node[data] (x) {\footnotesize SP};
\path (x)+(0.75,0) node[data] (y) {\footnotesize R0};
\path (x)+(0.375, -0.8) node[op,circle,inner sep=0.1cm] (add) {+};
\path (add)+(0.75,0) node[data] (z) {\footnotesize \texttt{4}};
\path (add)+(0.375, -0.8) node[op,circle,inner sep=0.1cm] (add1) {+};

\path[->] (x) edge[] node[] {} (add);
\path[->] (y) edge[] node[] {} (add);
\path[->] (add) edge[] node[] {} (add1);
\path[->] (z) edge[] node[] {} (add1);

\path (z)+(1,0) node [] (arrow1) {};
\path (arrow1.center)+(1,0) node[] (arrow2) {};
\path[->,line width=3pt] (arrow1.center) edge [] node[] {} (arrow2.center);

\path (arrow2)+(1,0.4) node [data] (x1) {\footnotesize SP};
\path (x1)+(0.75,0.2) node[data] (y1) {\footnotesize R0};
\path (y1)+(0.7,-0.3) node[data] (z1) {\footnotesize \texttt{4}};
\path (y1)+(-0.2,-1) node[op,circle,inner sep=0.1cm] (add2) {+};

\path[->] (x1) edge[] node[] {} (add2);
\path[->] (y1) edge[] node[] {} (add2);
\path[->] (z1) edge[] node[] {} (add2);
\end{tikzpicture}
}
\end{center}

\vspace*{-0.35cm}\paragraph{Miscellaneous translations}
Besides the rewrite rules described above, we apply additional miscellaneous rules that do not fit any of the aforementioned categories. They are listed in Appendix~\ref{appendix:misctrans}.





\subsection{Advantages}
\label{sec:dfgconstructadvantages}
Applying the normalization rewrite rules during construction of the graph has several advantages over doing so once the graph is fully generated. First, in case normalization function $\op$ has constant running time complexity, then the running time complexity of the construction phase, including normalization, grows linearly with the number of assembly instructions, whereas repeated application on a wholly generated DFG has quadratic complexity.

Second, by Lemma~\ref{lem:normalized}, 
equivalence of any pair of node references can be evaluated in constant time, simply by checking whether $v_1 = v_2$. As such, substitution of \texttt{LOAD} operations by their result can be achieved in constant time. The property is also utilized extensively during symbolic execution (Section~\ref{sec:symbolic}).
%
%
Suppose some predicate $P$ involves node $v_1 \in V$. Then, a condition involving $v_2 \in V$, can be evaluated immediately under $P$ without the need for proving equivalence of $v_1$ and $v_2$ first.

\section{Symbolic execution}\label{sec:symbolic}
During the analysis of a function, we may encounter conditional instructions. By definition, a conditional instruction carries a condition. We define the terms \emph{determined} and \emph{underdetermined} conditions. These terms relate to the terminology used in the classification of systems of linear equations. For \emph{determined} conditions, the input variables are restricted to a domain such that there is only a single possible evaluation result. For example, a conditional jump instruction at the end of a loop consisting of a fixed number of iterations. Conversely, for \emph{underdetermined} conditions, the input variables are not restricted enough to determine a fixed outcome. Below we describe how we approach this class of conditions.

During the DFG construction of any function \fu, we keep a state $\St = (G, P, B)$, where $G = (V,E)$ is the partially constructed DFG. $P$ is the path condition, which is constructed during symbolic execution; a predicate restricting unknown variables to a certain domain so that, if satisfied, the execution path follows the same path taken during the DFG construction. Phrased differently: satisfaction of $P$ warrants that $G$ represents the input/output relation of $\fu$. The inverse of this statement need not be true. 
Finally, backlog $B$ is a mapping between an execution address and a list of booleans. For all underdetermined conditional instructions encountered during the construction of $G$, $B$ keeps a record of which evaluation result was chosen (i.e., true/false). Since the analysis may encounter the same conditional instruction several times, a list is kept. We define $B_e[i] \in \mathbb{B}$, as the evaluation result chosen during the $i^{th}$ occurrence of the underdetermined conditional instruction located at execution address $e$.

The graph construction begins by initializing $\St = (G, P, B)$ to the empty state, i.e. $G$ is an empty graph, $P = \textbf{true}$, and $B$ has no record of any evaluation result. Then, we begin the construction by processing the instruction located at the entry point of function \fu. Some instructions may manipulate the execution flow, for e.g., a branch instruction, in which case, we continue at its target address. The construction is complete when we encounter an instruction causing the execution flow to return to \fu's calling function. For example, in ARM assembly, this is achieved by writing the initial value of register \texttt{LR}, as set by the caller of \fu, to the program counter register \texttt{PC}.


We represent a condition $c$ in the form of a tuple $(v_1, o, v_2)$, where $v_1, v_2 \in V$, and $o \in \{<, \leq, =, \geq, >\}$ is the operator.
In case either $v_1$ or $v_2$ is non-constant, $c$ need not be underdetermined, as predicate $P$ may sufficiently restrict $v_0$ or $v_1$ so that $c$ is determined.
In case $c$ is underdetermined, both execution paths are possible, and we are forced to choose which one to follow. Alternatively, we may follow both paths, by duplicating state \St, and subsequently assigning each execution path to one of the instances. This way, the resulting final graph construction consists of several DFGs; each one representing a different execution path. We refer to this practice as \emph{forking} state \St.
Forking at the occurrence of every underdetermined condition maximizes code coverage. However, it is infeasible due to the state explosion problem. Therefore, we should devise a balanced strategy for when to apply it -- as elaborated below.

\subsection{Path Oracle}\label{sec:pathoracle}
The strategy of when to apply forking only loosely relates to the symbolic execution itself. Therefore, we introduce the \emph{Path Oracle}, a separate entity that is queried during the graph construction phase, for every occurrence of an underdetermined condition $c$. It decides whether $c$ should evaluate to \textbf{true} or \textbf{false}, or that the construction should fork and follow both execution paths. 

\alglanguage{pseudocode}
\begin{algorithm}[h!tb]
\footnotesize
 \caption{Conditional Instruction
  \label{alg:conditionalinstruction}}
  \begin{algorithmic}[0]
   \Require{$\St = (G, P, B)$, ExecutionAddress e, Condition c, PathOracle po}
    \If{$P \land c = \textbf{true}$}
     \State Evaluate instruction at $e$
    \ElsIf{$P \land c = \textbf{false}$}
     \State Skip over instruction at $e$
    \Else
     \Let{d}{po.query(e, B)}
     \If{d = \texttt{TAKE\_TRUE}}
      \Let{$P$}{$P \land c$} \Comment{expand $P$ with $c$}
      \Let{$B_e$}{$B_e \cup \{\textbf{true}\}$} \Comment{append decision to backlog}
      \State Evaluate instruction at $e$
     \ElsIf{d = \texttt{TAKE\_FALSE}}
      \Let{$P$}{$P \land \neg{c}$}
      \Let{$B_e$}{$B_e \cup \{\textbf{false}\}$}
      \State Skip over instruction at $e$
     \ElsIf{d = \texttt{TAKE\_BOTH}}
      \Let{$\St'$}{\St.fork()} \Comment{$\St' = (G', P', B')$}
      \Let{$P$}{$P \land c$}
      \Let{$B_e$}{$B_e \cup \{\textbf{true}\}$}
      \Let{$P'$}{$P' \land \neg{c}$}
      \Let{$B_e'$}{$B_e' \cup \{\textbf{false}\}$}
      \State $e$ is evaluated for $\St$, skipped for $\St'$
     \EndIf
    \EndIf
  \end{algorithmic}
\end{algorithm}

For every decision made by the path oracle, $P$ and $B$ in $\St$ are updated accordingly. The pseudocode given in Algorithm~\ref{alg:conditionalinstruction} depicts how this is done. In short, predicate $P$ is updated to include condition $c$ (or the negation thereof), thereby maintaining satisfaction of its defining property, i.e. satisfaction of $P$ guarantees $G$ represents the input/output relation of $\fu$. An entry is added to backlog $B$, reflecting the decision made by the path oracle. $B$ has no purpose beyond weighing into the decisions made by the path oracle.


\subsubsection{Path Oracle Policy}\label{sec:pathoraclepolicy}
The goal of the policy described below is, for some number $n$, to obtain a DFG consisting of exactly $n$ iterations of a primitive with variable input length.
The target primitive can subsequently be identified by searching for exactly $n$ iterations in the resulting DFG.

We define $d_{e,i} \in \{\texttt{TAKE\_TRUE}, \texttt{TAKE\_FALSE}, \texttt{TAKE\_BOTH}\}$ as the path oracle's decision for the $i^{th}$ query for the conditional instruction found at execution address $e$. The policy for the path oracle is defined as follows:
\[
\arraycolsep=1.4pt
\begin{array}{rllll}
 d_{e,0} := & \multicolumn{2}{l}{$\texttt{TAKE\_BOTH}$} \\
 \multirow{2}{*}{$d_{e,i} := $} & \multirow{2}{*}{\scalebox{0.8}{$\biggl\{$}} & \texttt{TAKE\_TRUE}  & \text{ iff } B_e[0] = \textbf{true}, & \multirow{2}{*}{\scalebox{0.8}{$\biggr\}$}$~~\forall_i \in [1, n-1]$}\\
 & &  \texttt{TAKE\_FALSE} & \text{ iff } {B_e[0]} = \textbf{false} \\

 \multirow{2}{*}{$d_{e,i} := $} & \multirow{2}{*}{\scalebox{0.8}{$\biggl\{$}} & \texttt{TAKE\_FALSE} & \text{ iff } B_e[0] = \textbf{true}, & \multirow{2}{*}{\scalebox{0.8}{$\biggr\}$}$~~\forall_i \in [n, \infty]$}\\
 & & \texttt{TAKE\_TRUE} & \text{ iff } {B_e[0]} = \textbf{false}
\end{array}
\]

We justify the choice of policy by means of an example. Suppose that we encounter an underdetermined condition $c$ at address $e$. We do not know which of the two possible execution paths leads to a cryptographic primitive (if any). Hence, for $i=0$, i.e., the first occurrence, we fork the state and explore both. Suppose that, at a later point during the graph construction, one instance visits address $e$ again, hence $i=1$, and finds itself with another underdetermined condition $c'$. Since, at this point, $P$ incorporates $c$ (or $\neg c$), the outcome of $c$ can be evaluated. As $c'$ is underdetermined, $c \neq c'$ is guaranteed.

Such behavior is typical for a loop-guard statement. If this is indeed the case, the execution path taken at $i=0$ made us revisit $e$. In light of our goal of constructing a DFG comprising of $n$ iterations of a primitive, we replicate this path choice $n-1$ times, and subsequently take the opposite path, causing the execution flow to exit the loop.
Finally, the construction phase yields two DFGs: one representing $0$ iterations, and another representing $n$ iterations.
A description of the strategy being applied to a concrete example is given in Appendix~\ref{appendix:oraclePolicy}. The strategy does not produce exactly $n$ iterations in every situation. Section~\ref{sec:pathoraclelimitations} highlights typical exceptions.

\section{Purging process}\label{sec:dfgpurge}
Once the construction is complete, graph $G$ represents the input/output relation of $\fu$, under predicate $P$. 
However, it contains other information as well, such as nodes created from temporary loads/stores to the stack, and expressions rewritten by the broker, leaving the source nodes unused.
For e.g., suppose that $v$ represents $\texttt{ADD}(x,y)$. Then, a request to the broker for $\texttt{ADD}(v,z)$ yields node $w$, representing $\texttt{ADD}(x,y,z)$. $w$ does not depend on $v$ and, unless $v$ is referenced independently elsewhere, $v$ is not part of $\fu$'s input/output relation.

Leaf nodes are, by definition, graph nodes that are not used as an input to any arithmetic/logical operation.
Our approach becomes the following: for each leaf node $v$, we check whether it is part of $\fu$'s semantics. We consider leaf node $v$ to be part of $\fu$'s semantics, if $v$ is either:
\begin{enumerate}[(i)]
\vspace{-1mm}
 \item the return value of $\fu$, 
 \vspace{-2mm}
 \item a \texttt{STORE} operation, and the target address is not relative to the \texttt{SP} register. Thus, information is stored outside of the stack, or
  \vspace{-2mm}
 \item a \texttt{CALL} operation, i.e. a function call not subject to inlining.
 \vspace{-1mm}
\end{enumerate}
In case none of the above applies, $v$ and its incoming edges can be removed from $G$, without affecting its semantics. The removal of leaf nodes continues repeatedly until no more nodes can be removed. Finally, by construction, all nodes in $G$ are either leaf nodes that are part of $\fu$'s semantics, or intermediate results contributing to some leaf.

\section{Signature Expression}\label{sec:dfgsearch}


\begin{figure}[htb]
 \vspace*{-1.5cm}
 \begin{subfigure}{0.550\columnwidth}
  \centering
  \vspace{2.25cm}
  \scalebox{0.6}{\begin{tikzpicture}[auto]
\node[parseinit] (start) {};
\path (start)+(0.25,0) node[] (junc0) {};
\path (junc0.center)+(1.2,-0.5) node[parse] (id) {\small \textsc{identifier}};
\path[->] (start) edge[] node[] {} (junc0.center);
\path (id)+(1.5,0) node [parseop] (sigid) {\scriptsize string};
\path[->] (junc0.center) edge[in=180,out=0,looseness=1.0] node[] {} (id);
\path[->] (id) edge[] node[] {} (sigid);
\path (sigid.east)+(0.3,0.5) node[] (junc1) {};
\path[->] (junc0.center) edge[] node[] {} (junc1.center);
\path[->] (sigid) edge[in=180,out=0,looseness=1.0] node[] {} (junc1.center);

\path (junc1.center)+(1.05,-0.5) node[parse] (var) {\small \textsc{variant}};
\path (var)+(1.35,0) node [parseop] (sigvar) {\scriptsize string};
\path[->] (junc1.center) edge[in=180,out=0,looseness=1.0] node[] {} (var);
\path[->] (var) edge[] node[] {} (sigvar);
\path (sigvar.east)+(0.3,0.5) node[] (junc2) {};
\path[->] (junc1.center) edge[] node[] {} (junc2.center);
\path[->] (sigvar) edge[in=180,out=0,looseness=1.0] node[] {} (junc2.center);

\path (junc2)+(0.5,0) node[] (junc26) {};
\path[-] (junc2.center) edge[] node[] {} (junc26.center);
\path (junc26.center)+(0.15, -0.15) node[] (junc3) {};
\path[-] (junc26.center) edge[in=90,out=0,looseness=1.0] node[] {} (junc3.center);
\path (junc3.center)+(0,-0.75) node[] (junc4) {};
\path[-] (junc3.center) edge[] node[] {} (junc4.center);
\path (junc4.center)+(-0.15,-0.15) node[] (junc5) {};
\path[-] (junc4.center) edge[in=0,out=270,looseness=1.0] node[] {} (junc5.center);
\path (junc0.center)+(0.3,-1.05) node[] (junc6) {};
\path (junc6)+(3.5,0) node[] (junc28) {};
\path[->] (junc5.center) edge[] node[] {} (junc28.center);
\path[-] (junc28.center) edge[] node[] {} (junc6.center);
\path (junc6.center)+(-0.15,-0.15) node[] (junc7) {};
\path[-] (junc6.center) edge[in=90,out=180,looseness=1.0] node[] {} (junc7.center);
\path (junc7.center)+(0,-0.30) node[] (junc8) {};
\path[-] (junc7.center) edge[] node[] {} (junc8.center);
\path (junc8.center)+(0.15,-0.15) node[] (junc9) {};
\path[->] (junc8.center) edge[in=180,out=270,looseness=1.0] node[] {} (junc9.center);

\path (junc9.center)+(1.2,-0.5) node[parse] (transient) {\small \textsc{transient}};
\path (transient.east)+(0.3,0.5) node[] (junc10) {};
\path[->] (junc9.center) edge[in=180,out=0,looseness=1.0] node[] {} (transient);
\path[->] (junc9.center) edge[] node[] {} (junc10.center);
\path[->] (transient) edge[in=180,out=0,looseness=1.0] node[] {} (junc10.center);
\path (junc10.center)+(0.7,-0.5) node[parseop] (label) {\scriptsize label};
\path[->] (junc10.center) edge[in=180,out=0,looseness=1.0] node[] {} (label);
\path (label)+(0.85,0) node [parse] (colon) {\small \textsc{:}};
\path[->] (label) edge[] node[] {} (colon);
\path (colon.east)+(0.25,0.5) node[] (junc11) {};
\path[->] (junc10.center) edge[] node[] {} (junc11.center);
\path[->] (colon) edge[in=180,out=0,looseness=1.0] node[] {} (junc11.center);

\path (junc11)+(0.85,0) node [parseop] (expr) {\scriptsize expression};
\path[-] (junc11.center) edge[] node[] {} (expr);
\path (expr)+(1.15,0) node [parse] (semicolon) {\small \textsc{;}};
\path[->] (expr) edge[] node[] {} (semicolon);

\path (semicolon.east)+(0.15, -0.15) node[] (junc12) {};
\path[-] (semicolon) edge[in=90,out=0,looseness=1.0] node[] {} (junc12.center);
\path (junc12.center)+(0,-0.75) node[] (junc13) {};
\path[-] (junc12.center) edge[] node[] {} (junc13.center);
\path (junc13.center)+(-0.15,-0.15) node[] (junc14) {};
\path[-] (junc13.center) edge[in=0,out=270,looseness=1.0] node[] {} (junc14.center);
\path (junc9.center)+(0,-1.05) node[] (junc15) {};
\path (junc15)+(3.5,0) node[] (junc27) {};
\path[->] (junc14.center) edge[] node[] {} (junc27.center);
\path[-] (junc27.center) edge[] node[] {} (junc15.center);

\path (junc15.center)+(-0.15,-0.15) node[] (junc16) {};
\path[-] (junc15.center) edge[in=90,out=180,looseness=1.0] node[] {} (junc16.center);
\path (junc16.center)+(0,-0.30) node[] (junc17) {};
\path[-] (junc16.center) edge[] node[] {} (junc17.center);
\path (junc17.center)+(0.15,-0.15) node[] (junc18) {};
\path[-] (junc17.center) edge[in=180,out=270,looseness=1.0] node[] {} (junc18.center);

\path (junc15.center)+(-0.15,0.15) node[] (junc24) {};
\path[-] (junc15.center) edge[in=270,out=180,looseness=1.0] node[] {} (junc24.center);
\path (junc9.center)+(-0.15,-0.15) node[] (junc25) {};
\path[-] (junc25.center) edge[in=180,out=90,looseness=1.0] node[] {} (junc9.center);
\path[-] (junc24.center) edge[] node[] {} (junc25.center);

\path (junc18.center)+(0.875,0) node[parse] (var2) {\small \textsc{variant}};
\path (var2)+(1.35,0) node [parseop] (sigvar2) {\scriptsize string};
\path[->] (junc18.center) edge[] node[] {} (var2);
\path[->] (var2) edge[] node[] {} (sigvar2);

\path (junc5.center |- sigvar2) node [] (junc19) {};
\path[-] (sigvar2) edge[] node[] {} (junc19.center);
\path (junc19.center)+(0.15,0.15) node[] (junc20) {};
\path[-] (junc19.center) edge[in=270,out=0,looseness=1.0] node[] {} (junc20.center);
\path (junc5)+(0.15,-0.15) node[] (junc21) {};
\path (junc21)+(-0.15,0.15) node[] (junc26) {};
\path[-] (junc20.center) edge[] node[] {} (junc21.center);
\path[->] (junc21.center) edge[in=0,out=90,looseness=1.0] node[] {} (junc5.center);

\path (junc17.center)+(0,-0.60) node[] (junc22) {};
\path[-] (junc17.center) edge[] node[] {} (junc22.center);
\path (junc22.center)+(0.15,-0.15) node[] (junc23) {};
\path[-] (junc22.center) edge[in=180,out=270,looseness=1.0] node[] {} (junc23.center);

\path (junc23.center)+(0.25,0) node[parseinit] (end) {};
\path[->] (junc23.center) edge[] node[] {} (end);
\end{tikzpicture}}
  \caption{High-level state machine}
  \label{fig:parser}
 \end{subfigure}
 \begin{subfigure}{0.425\columnwidth}
  \scalebox{0.6}{\input{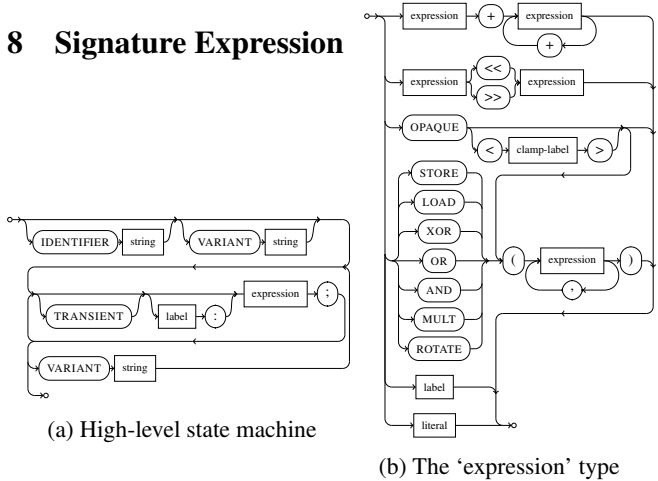}}
  \caption{The `expression' type}
  \label{fig:expression}
 \end{subfigure}
 \caption{Diagram representation of the DSL parser}
\end{figure}

\noindent In order to detect subgraph isomorphism, we need a means of expressing the signature graph. Figure~\ref{fig:parser} depicts a diagram of the signature domain-specific language (DSL). 
Appendix~\ref{appendix:signature} provides a concrete example.
The round boxes denote a keyword, whereas the square boxes denote a data type. New graph nodes are generated through the \emph{expression} data type (Figure~\ref{fig:expression}).
The \textsc{identifier} keyword allows one to specify a friendly name for the signature.
The \textsc{variant} keyword enforces the creation of a new empty DFG. Subsequent expressions are added to this graph, thus, allowing one to specify multiple variants of a signature. Subgraph isomorphism detection is ultimately performed with all variants. 
The \emph{label} data type is an optional field. It allows the node to be referenced by another expression, enabling node sharing between expressions.
Analogous to DFGs generated from assembly instructions, a DFG declared in the DSL is also subject to normalization by the broker (Section~\ref{sec:dfgconstruct}), and purging (Section~\ref{sec:dfgpurge}). In case the \textsc{transient} keyword is specified, the node generated from the expression is considered to be non-essential, and may be removed during the purging process (i.e. in case it was translated by the broker).

Figure~\ref{fig:expression} depicts the \emph{expression} data type. It is recursively defined, and hence allows for nested subexpressions.
The `\textsc{+}' keyword denotes the addition of two or more subexpressions. `\textsc{\textless{}\textless{}}'~/~`\textsc{\textgreater{}\textgreater{}}' denote a left and right shift, respectively. The \emph{label} data type is a reference to a previously defined graph node. The \emph{literal} data type denotes a constant value. The \textsc{store}, \textsc{load}, \textsc{xor}, \textsc{or}, \textsc{and}, \textsc{mult} and \textsc{rotate} keywords followed by subexpressions contained in parentheses provoke creation of a new graph node. The subexpressions serve as input nodes.
Finally, the \textsc{opaque} keyword signifies a special wildcard node. A comparison with a node of any other type by the subgraph-isomorphism algorithm always yields true. The opaque node type can have any number of input nodes, including zero. 
The optional \emph{clamp-label} data type allows one to assign a name to the node type. Consequently, a comparison with a node of any other type yields true, with the added restriction that all opaque nodes carrying the same type label must map to nodes of the same type. We refer to this practice as \emph{type clamping}.

Within the realm of identifying unknown primitives, 
a special wildcard applicable to a \emph{group} of nodes would be useful. However, to our knowledge, the nature of subgraph-isomorphism does not allow for the augmentation of any such algorithm to support one-to-many mappings. 
Alternatively, one may declare several variants of a signature, where for each variant, the wildcard group is denoted by a different number of nested opaque operations, i.e. \texttt{OPAQUE}, \texttt{OPAQUE(OPAQUE)},
etc. This way, any group consisting of a finite number of operations can be expressed.
Introducing a notation triggering the translation to multiple variants automatically has been considered. However, as the number of signature variants grows exponentially in the usage count of this hypothetical notation, we prefer to discourage its use. Hence, we omit the notation altogether, enforcing explicit declaration of multiple variants.


\section{Subgraph isomorphism}

Subgraph isomorphism is a well-documented problem, and is known to be NP complete. The solution proposed by Ullmann~\cite{ullmann1976} is a recursive backtracking algorithm with pruning. Our framework implements this algorithm, with added support for type clamping (see Section~\ref{sec:dfgsearch}). 
For further details about Ullmann's algorithm and the optimizations we applied to it, we refer the reader to the documentation included with our framework's source code.


\section{Signatures}\label{sec:signatures}
Before diving into the practical performance evaluation, we highlight the signatures used throughout the analysis, along with relevant details and a motivation as to why they are included. All signature definition files are included in our implementation of the framework. The list given below should not be interpreted as an attempt to cover the entirety of cryptographic primitives in existence. Rather, they showcase the applicability of our framework. The selection of signatures was made with a strong focus on proprietary algorithms in embedded environments. As such, they consist of symmetric and unkeyed primitives only, although there is no fundamental incompatibility with asymmetric primitives.
To our knowledge, no proprietary primitive exists to date that is studied in the scientific literature and does not fall within any of the classes covered in this section.

However, should an additional signature be desired, then it can be crafted. In broad terms, the approach is to formulate the primitive's defining properties, translate those to an abstract DFG, and finally into a signature definition expressed in the DSL. The process is somewhat ad-hoc in nature. However, the examples presented this section should provide sufficient guidance.

\subsection{AES, MD5, XTEA, SHA1} 

Despite this paper's strong focus on unknown primitives, and hence generic signatures, algorithm-specific signatures, such as AES, MD5, XTEA and SHA1, can be defined and used. Doing so allows us to directly compare results with~\cite{lestringant2015automated}, and demonstrate that our approach effectively solves the code fragment selection problem without resorting to heuristics.

\begin{wrapfigure}{r}{0.32\columnwidth}
\vspace*{-1.3cm}
\scalebox{0.8}{
\begin{tikzpicture}[auto]
\node[data] (l0) {\footnotesize $L_0$};
\path (l0)+(3,0) node[data] (r0) {\footnotesize $R_0$};
\path (l0)+(0,-1.5) node[op,circle,inner sep=0.05cm] (xor) {\footnotesize \textsc{xor}};
\path (xor)+(1.5,0) node[opaque] (f) {\footnotesize $F$};
\path (f)+(1.5,0) node[] (junc) {};
\path (f)+(0, 1) node[data] (k0) {\footnotesize $K_0$};
\path (junc.center |- xor.south) node[] (r00) {};

\path[->] (l0) edge[] node[] {} (xor);
\path[->] (f) edge[] node[] {} (xor);
\path[-] (r0) edge[] node[] {} (junc.center);
\path[->] (junc.center) edge[] node[] {} (f);
\path[-] (junc.center) edge[] node[] {} (r00.center);
\path[->] (k0) edge[] node[] {} (f);

\path (xor)+(0,-2.5) node[op,circle,inner sep=0.05cm] (xor1) {\footnotesize \textsc{xor}};
\path (xor1)+(1.5,0) node[opaque] (f1) {\footnotesize $F$};
\path (f1)+(1.5,0) node[] (junc1) {};
\path (f1)+(0, 1) node[data] (k1) {\footnotesize $K_1$};
\path (k1)+(-1.5,-0.2) node[] (l1) {};
\path (k1)+(1.5,-0.2) node[] (r1) {};
\path (junc1.center |- xor1.south) node[] (r10) {};

\path[-] (r00.center) edge[in=90,out=270,looseness=0.8] node[] {} (l1.center);
\path[-] (xor) edge[in=90,out=270,looseness=0.8] node[] {} (r1.center);
\path[->] (l1.center) edge[] node[] {} (xor1);
\path[-] (r1.center) edge[] node[] {} (junc1.center);
\path[->] (junc1.center) edge[] node[] {} (f1);
\path[-] (junc1.center) edge[] node[] {} (r10.center);
\path[->] (k1) edge[] node[] {} (f1);
\path[->] (f1) edge[] node[] {} (xor1);

\path (xor1)+(1.5, -1.5) node[] (k2) {$\vdots$};
\path (k2)+(-1.5,-0.2) node[] (l2) {};
\path (k2)+(1.5,-0.2) node[] (r2) {};

\path[->] (r10.center) edge[in=90,out=270,looseness=0.8] node[] {} (l2.center);
\path[->] (xor1) edge[in=90,out=270,looseness=0.8] node[] {} (r2.center);
\end{tikzpicture}
}
\caption{DFG of a Feistel structure 
}
  \label{fig:feistel} 
  \vspace*{-0.5cm}
\end{wrapfigure}
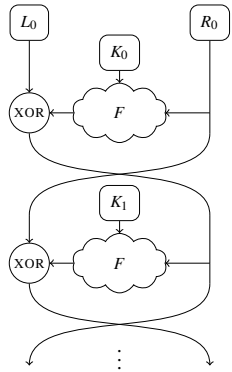

\subsection{Feistel cipher}\label{sec:feistel} A Feistel cipher is a symmetric structure used in many block ciphers, including DES. 
In a Feistel cipher, a plaintext block $P$ is split in two pieces $L_0$ and $R_0$.
Then, for each round $i \in [0, 1 , \dots , n]$,\\
$\begin{array}{ll}
~~~~~ & L_{i+1} = R_i\,\\
& R_{i+1} = L_i \oplus F(R_i, K_i),
\end{array}$\\
is computed, where $\oplus$ denotes bitwise exclusive-or, $F$ the round function, and $K_i$ the sub-key for round $i$. Translating this definition into a DFG yields the graph shown in Figure~\ref{fig:feistel}.

The next step is to construct a signature that represents the DFG from Figure~\ref{fig:feistel}. However, $F$ is an algorithm-specific set of operations, of which thus no properties are known. 
The \texttt{OPAQUE} operator (see Section~\ref{sec:dfgsearch}), only covers a single operation, whereas $F$ consists of an unknown number of operations. 
$F$ is known to take $R_i$ and $K_i$ as an input, where $i \in [0, 1, \dots, n]$. No properties are known for $K_i$. Hence, we represent $F$ by introducing multiple variants of the signature. In the first variant, we substitute $F$ with $\texttt{OPAQUE}(R_i)$, in the second with $\texttt{OPAQUE}(\texttt{OPAQUE}(R_i))$, etc., until we reach $8$ levels of nested operations. Thus, the signature identifies Feistel ciphers with an $F$ whose input/output relation contains between $1$ and $8$ successive operations.

\subsection{(Non-)Linear feedback shift register} (Non-)Linear feedback shift registers ((N)LFSRs) are often used in pseudo-random number generators, and key-stream generators for stream ciphers. When designed carefully, an (N)LFSR offers relatively strong randomness, whilst requiring very few logic gates, often making it an attractive choice for algorithms used in embedded devices. Both hardware and software implementations of (N)LFSRs are common.

Let $R$ be an (N)LFSR. For each round, a new bit is generated using feedback function $L$ from (a subset of) the bits in $R$. If $L$ is linear, for e.g. an exclusive-or over the input bits, we refer to $R$ as an LFSR. Conversely, $R$ is an NLFSR if $L$ is non-linear. All bits in register $R$ are shifted one position to the left, discarding the most significant bit, and the newly generated bit is placed at position $0$. 
Furthermore, an output bit is generated by feeding $R$ to some function $F$.
Hence, we have,
for each round $i \in [0, 1, \dots, n]$, \\[-0.35cm]
\begin{wrapfigure}{r}{0.4\columnwidth}
\scalebox{0.8}{
\begin{tikzpicture}[auto]
\node[data] (r0) {\footnotesize $R_0$};
\path (r0)+(2,0) node[opaque] (F) {\footnotesize $F$};
\path (r0)+(2,-1) node[data] (one) {\footnotesize 1};
\path (r0)+(1,-1) node[op,circle,inner sep=0.1cm] (lsh) {\footnotesize \textsc{\textless{}\textless{}}};
\path (lsh)+(-1,-1) node[op,circle,inner sep=0.1cm] (r1) {\footnotesize \textsc{or}};
\path (r0)+(-1,-1) node[opaque] (L) {\footnotesize $L$};

\path[->] (r0) edge[out=195,in=75] node[] {} (L);
\path[->] (r0) edge[] node[] {} (F);
\path[->] (r0) edge[out=345,in=105] node[] {} (lsh);
\path[->] (one) edge[] node[] {} (lsh);
\path[->] (L) edge[out=285,in=165] node[] {} (r1);
\path[->] (lsh) edge[out=255,in=15] node[] {} (r1);

\path (r1)+(2,0) node[opaque] (F1) {\footnotesize $F$};
\path (r1)+(2,-1) node[data] (one1) {\footnotesize 1};
\path (r1)+(1,-1) node[op,circle,inner sep=0.1cm] (lsh1) {\footnotesize \textsc{\textless{}\textless{}}};
\path (lsh1)+(-1,-1) node[op,circle,inner sep=0.1cm] (r2) {\footnotesize \textsc{or}};
\path (r1)+(-1,-1) node[opaque] (L1) {\footnotesize $L$};

\path[->] (r1) edge[out=195,in=75] node[] {} (L1);
\path[->] (r1) edge[] node[] {} (F1);
\path[->] (r1) edge[out=345,in=105] node[] {} (lsh1);
\path[->] (one1) edge[] node[] {} (lsh1);
\path[->] (L1) edge[out=285,in=165] node[] {} (r2);
\path[->] (lsh1) edge[out=255,in=15] node[] {} (r2);

\path (r2)+(2,0) node[opaque] (F2) {\footnotesize $F$};
\path (r2)+(0.9,-0.6) node[] (fade) {};
\path (r2)+(-0.9,-0.6) node[] (fade1) {};
\path (r2)+(0,-0.6) node[] (dots) {$\vdots$};

\path[->] (r2) edge[] node[] {} (F2);
\path[->] (r2) edge[out=195,in=75] node[] {} (fade1.center);
\path[->] (r2) edge[out=345,in=105] node[] {} (fade.center);
\end{tikzpicture}
}
\caption{DFG of an (N)LFSR}
  \label{fig:lfsr}
  \vspace{-0.33cm}
\end{wrapfigure}
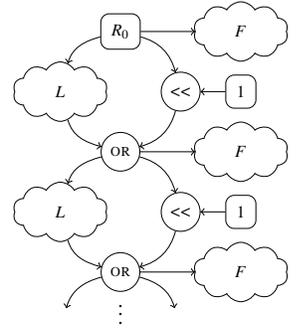
$\arraycolsep=1.4pt\begin{array}{rl}
R_{i+1} &= (R_i \text{~\textless{}\textless{}~} 1)~|~L(R_i)\\
~~~\text{output}_i &= F(R_i),
\end{array}$\\
where $\text{\textless{}\textless{}~} x$ denotes a left shift by $x$ bits and $|$ denotes bitwise or.

Figure~\ref{fig:lfsr} depicts a translation of the above into a DFG. In order to express this graph in a signature, we replace $L$ and $F$ with \texttt{OPAQUE} operators. The property that $R_{i+1}$ depends on $R_i$ via $L$ is lost. However, the signature remains distinctive enough in order to warrant very few false positives (see Section~\ref{sec:evaluation}).

\subsection{Sequential Block Permutation}\label{sec:merkledamgard}
Variable-length primitives constructed from fixed-length ones are a common phenomenon. For e.g., all hash functions built on the Merkle-Damg{\aa}rd construction, such as MD5, SHA1 and SHA2, have this characteristic. Other examples include block ciphers in a chaining mode of operation. We refer to this concept as a \emph{sequential block permutation}.

Let $H_i$ be the $i^{\text{th}}$ output block of a sequential block permutation function, $B_i$ be the $i^{\text{th}}$ input block, $c$ be the fixed-length compression function, for $i \in [0, 1, \dots, n]$. $I$ denotes the initialization vector. Then, we define the sequential block permutation as:\\
$\arraycolsep=1.4pt
\begin{array}{rlp{0.5cm}l}
~~~~~~~~~~~~~~~~~~~~~~H_0 & = c(I, B_0)\\
H_i & = c(H_{i-1}, B_i) & &  \forall_i \in [1, n]
\end{array}
$


\noindent A DFG representation is given in Figure~\ref{fig:merkledamgard}. On inspection, we find that it only provides structural guidance, and does not prescribe any arithmetic or logic operations.
The definition of $H$ prescribes that compression function $c$ takes two inputs:
\begin{enumerate}[(i)]
\vspace{-1mm}
 \item The output of its preceding instance, except for the first instance, which depends on the IV.
 \vspace{-2mm}
 \item Any of the input blocks $B_0, B_1, \dots, B_n$.
 \vspace{-1mm}
\end{enumerate}

\begin{wrapfigure}{r}{0.31\columnwidth}
\scalebox{0.8}{
\begin{tikzpicture}[auto]
\node[data] (i) {\footnotesize $I$};
\path (i)+(2,-1) node[data] (b0) {\footnotesize $B_0$};
\path (i)+(0,-1) node[opaque] (c) {\footnotesize $c$};

\path[->] (i) edge[] node[] {} (c);
\path[->] (b0) edge[] node[] {} (c);

\path (c)+(2,-1.1) node[data] (b1) {\footnotesize $B_1$};
\path (c)+(0,-1.1) node[opaque] (c1) {\footnotesize $c$};

\path[->] (c) edge[] node[] {} (c1);
\path[->] (b1) edge[] node {} (c1);

\path (c1)+(2,-1.1) node[data] (b2) {\footnotesize $B_2$};
\path (c1)+(0,-1.1) node[opaque] (c2) {\footnotesize $c$};

\path[->] (c1) edge[] node[] {} (c2);
\path[->] (b2) edge[] node {} (c2);

\path (c2)+(0,-0.55) node[] (junc) {};
\path[->] (c2) edge[] node[] {} (junc.center);

\path (c2)+(0,-0.8) node[] (k2) {$\vdots$};

\path (junc.center)+(0,-0.55) node[] (junc1) {};
\path (junc1.center)+(0,-0.55) node[opaque] (cn) {\footnotesize $c$};
\path (cn)+(2,0) node[data] (bn) {\footnotesize $B_n$};

\path[->] (junc1.center) edge[] node[] {} (cn);
\path[->] (bn) edge[] node[] {} (cn);

\path[->] (b0) edge[bend left=15,cyan,dashed] node[] {} (c);
\path[->] (c) edge[bend left=15,cyan,dashed] node[] {} (c1);
\path[->] (c1) edge[bend left=15,cyan,dashed] node[] {} (b1);

\path[->] (b1) edge[bend left=15,cyan,dashed] node[] {} (c1);
\path[->] (c1) edge[bend left=15,cyan,dashed] node[] {} (c2);
\path[->] (c2) edge[bend left=15,cyan,dashed] node[] {} (b2);

\end{tikzpicture}
}
\caption{DFG of a sequential block permutation. The blue arrows depict the visitation order by the classifier}
  \label{fig:merkledamgard} 
\end{wrapfigure}
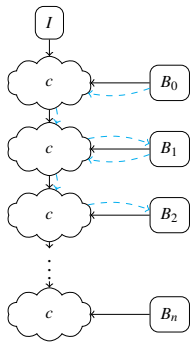
In order to express this in a signature definition, we may opt for an approach similar to how the Feistel cipher signature definition is constructed. However, Figure~\ref{fig:merkledamgard} does not contain any operation that serves as an `anchor point' for $c$, analogous to the XOR-operation in the Feistel structure. As such, any pattern of repeated operations satisfies property (i), which is overtly generic. Hence, we must also take property (ii) into account.
Let $c_i$ be the $i^{\text{th}}$ instance of $c$. The number of arithmetic/logical operations on the path between $c_{i-1}$ and $c_i$ need not be related to that of the path between input block $B_i$ and $c_i$.
Therefore, in order to translate $c$ into multiple variants of the signature, we have to perform a translation for both paths independently. Note that the number of variants grows exponentially in the number of translations. On top of that, the compression function $c$ can be vastly more complex than a round function in a Feistel cipher. For e.g., the MD5 compression function in itself consists of $64$ rounds. Therefore, the upper bound of the number of operations that $c$ may consist of is an order of magnitude higher than what one would typically find in a Feistel cipher's round function. All in all, the number of signature variants, and therewith the running time of the analysis, becomes prohibitively large.

Fortunately, there is no need to restrict ourselves to subgraph isomorphism as a means of identifying primitives. Rather, we can apply any algorithm to the DFGs generated by the graph construction framework, which is our approach for the sequential block permutation use case.
We take several observations into account.
First, input blocks $B_0, B_1, \dots, B_n$ are typically loaded from a memory address. Second, $c$ has a fixed (unknown) block size, and thus we can safely assume that the offsets between the load addresses of $B_i$, $B_{i+1}$ and $B_{i+2}$ are constant.
We take the following approach:
\begin{enumerate}[(i)]
\vspace{-1mm}
 \item We identify all nodes representing $\texttt{LOAD}(\texttt{ADD}(x, k))$, where $x$ is an arbitrary graph node, and $k$ is a constant. For each instance of $x$, we construct a list of tuples $(v_0, v_1, v_2)$, where $v_i$ represents $\texttt{LOAD}(\texttt{ADD}(x, k_i))$. A tuple is valid only if $k_1 - k_0 \geq 16 \land k_1 - k_0 = k_2 - k_1$, i.e. the offsets between $v_0, v_1$ and $v_2$ are constant, and at least $16$ bytes. As such, a DFG generated from a sequential block permutation function yields at least one tuple such that $v_i$ maps to $B_i$, for all $i \in [0, 1, 2]$.
\vspace{-2mm}
 \item For all tuples, we determine the shortest path between $v_0$ and $v_1$. This can be done by means of a simple breadth-first search. If $v_0$ maps to $B_0$ and $v_1$ to $B_1$, then this path should take us through two instances of $c$ (see Figure~\ref{fig:merkledamgard}).%
\vspace{-2mm}\item Suppose that such a path exists, then we would like to confirm that a similar path exists between $v_1$ and $v_2$. We take $v_1$ as a starting point, and traverse paths with edge directions and node types resembling those on the path between $v_0$ and $v_1$. Once such a path has been found, it should reach $v_2$. Satisfaction of this property is a strong positive indicator.
 \vspace{-2mm}
 \item To gain more certainty, we also verify that the node types of all inputs and outputs for all the nodes on both paths match. However, in case $v_0$ maps to $B_0$, some inputs may originate from the IV, whereas they originate from computed values during the second round. Therefore, we treat constants and inputs of type \texttt{LOAD} as wildcards in this step.
 \vspace{-1mm}
\end{enumerate}

\section{Experimental evaluation}\label{sec:evaluation}
We evaluate our solution's performance with regards to accuracy and running time on the following four test sets: (a) the sample set used in~\cite{lestringant2015automated}, (b) a collection of shared libraries and executables part of the OpenWRT\footnote{\url{https://openwrt.org/docs/techref/targets/mvebu}} network equipment firmware, (c) a collection of proprietary cipher implementations built from public sources, and (d) a collection of real-world embedded firmwares  
(PLCs, ECUs). The evaluation is conducted on an AMD Ryzen 3600 machine with 16 GB of RAM, which is considered mid-range hardware nowadays. 


While not containing proprietary cryptography, the OpenWRT project is publicly available without legal issues around redistribution, contrary to firmwares which do. As such, this evaluation benefits the reproducibility of our work, as well as demonstrates the general principle, accuracy and performance on a test set representative of high-end embedded device firmware.
Given the uncertainty over the legality of redistribution, we refer to the original sources of the proprietary cipher implementations rather than publish our binary test set. Due to copyright restrictions, we unfortunately lack permission to publish the real-world embedded firmwares.

Section~\ref{sec:pathoraclepolicy} defines a tunable variable $n$, the target number of instances of an algorithm contained within a DFG. The value chosen for $n$ should be low as it correlates with the size of the constructed DFGs, and hence running time, but high enough so that all signatures listed in Section~\ref{sec:signatures} can be identified. The algorithm-specific and Feistel classifiers only target a single instance of an algorithm, and hence are not affected by $n$. Conversely, the (N)LFSR and sequential block permutation classifiers are, as they identify a primitive based on multiple instances. The latter identifies two successive instances of some unknown compression function $c$. Because the rewrite rules are designed to promote numeric simplification (Section~\ref{sec:dfgconstruct}), the initialization and finalization step of an algorithm may become merged with the first and last instance of $c$, respectively. Thus, by choosing $n=4$, the presence of two successive instances of $c$ in the DFG is warranted. Choosing a value beyond $4$ clearly does not offer any advantages regarding this property. Furthermore, identifying $4$ successive rounds of an (N)LSFR in a DFG produced from code that does not actually implement one is highly unlikely. Therefore, for the remainder of this section, we take $n=4$.


\subsection{Comparison with Lestringant et al.}
Lestringant et al.~\cite{lestringant2015automated} showcase the effectiveness of their method by successfully identifying AES, MD5 and XTEA in binary files. Unfortunately, their sample set was never published, and is compiled for x86, which our implementation currently does not support. Therefore, we constructed a new sample set for the ARM architecture that is as faithful as possible to theirs. The algorithms are taken from the cited sources\footnote{\url{https://en.wikipedia.org/w/index.php?title=XTEA}}\textsuperscript{,}\footnote{\url{https://tools.ietf.org/html/rfc1321}}\textsuperscript{,}\footnote{\url{https://github.com/BrianGladman/AES}}, and subsequently compiled with GCC 9.3.0, Clang 9.0.8, and MSVC 19.16 on all available optimization levels (O0--O3, debug/release). We use algorithm-specific signatures in order to warrant a fair comparison. The results are depicted in Table~\ref{tab:lestringanttable}. They show that all samples are identified successfully by (a variant of) their corresponding signatures, regardless of compiler and optimization level. This effectively demonstrates that our approach is equally capable of identifying these algorithms, without resorting to heuristics for fragment selection.

\begin{table}[t!]
\begin{threeparttable}[b]
\centering
\vspace{-0.3cm}
\scalebox{0.75}{\begin{tabular}{c c c c c c}
Signature & \hspace{-0.15cm}Compiler\hspace{-0.15cm} & \shortstack{-O0 /\\Debug} & -O1 & \shortstack{-O2 /\\Release} & -O3\\
\hline\\[-0.3cm]
\multirow{3}{*}{\shortstack{\textbf{XTEA}\\\footnotesize{4 rounds}\\\footnotesize{\hspace{-0.15cm}70 vertices\hspace{-0.15cm}}}} & \small{GCC} &\multicolumn{1}{l}{\hspace{0.2cm}ok \scriptsize{\textcolor{darkgray}{ ($1\text{ms}$)}}}\hspace{-0.4cm} & \multicolumn{1}{l}{\hspace{0.2cm}ok \scriptsize{\textcolor{darkgray}{ ($2\text{ms}$)}}}\hspace{-0.4cm} & \multicolumn{1}{l}{\hspace{0.2cm}ok \scriptsize{\textcolor{darkgray}{ ($2\text{ms}$)}}}\hspace{-0.4cm} & \multicolumn{1}{l}{\hspace{0.2cm}ok \scriptsize{\textcolor{darkgray}{ ($2\text{ms}$)}}}\hspace{-0.4cm}\\
 & \small{Clang} &\multicolumn{1}{l}{\hspace{0.2cm}ok \scriptsize{\textcolor{darkgray}{ ($1\text{ms}$)}}}\hspace{-0.4cm} & \multicolumn{1}{l}{\hspace{0.2cm}ok \scriptsize{\textcolor{darkgray}{ ($2\text{ms}$)}}}\hspace{-0.4cm} & \multicolumn{1}{l}{\hspace{0.2cm}ok \scriptsize{\textcolor{darkgray}{ ($2\text{ms}$)}}}\hspace{-0.4cm} & \multicolumn{1}{l}{\hspace{0.2cm}ok \scriptsize{\textcolor{darkgray}{ ($2\text{ms}$)}}}\hspace{-0.4cm}\\
 & \small{MSVC} &\multicolumn{1}{l}{\hspace{0.2cm}ok \scriptsize{\textcolor{darkgray}{ ($1\text{ms}$)}}}\hspace{-0.4cm} &  - & \multicolumn{1}{l}{\hspace{0.2cm}ok \scriptsize{\textcolor{darkgray}{ ($2\text{ms}$)}}}\hspace{-0.4cm} & - \\
\multirow{3}{*}{\shortstack{\textbf{MD5}\\\footnotesize{64 rounds}\\\footnotesize{\hspace{-0.15cm}458-618 vertices\hspace{-0.15cm}}}} & \small{GCC} &\multicolumn{1}{l}{\hspace{0.2cm}ok \scriptsize{\textcolor{darkgray}{ ($267\text{ms}$)}}}\hspace{-0.4cm} & \multicolumn{1}{l}{\hspace{0.2cm}ok \scriptsize{\textcolor{darkgray}{ ($335\text{ms}$)}}}\hspace{-0.4cm} & \multicolumn{1}{l}{\hspace{0.2cm}ok \scriptsize{\textcolor{darkgray}{ ($345\text{ms}$)}}}\hspace{-0.4cm} & \multicolumn{1}{l}{\hspace{0.2cm}ok \scriptsize{\textcolor{darkgray}{ ($348\text{ms}$)}}}\hspace{-0.4cm}\\
 & \small{Clang} &\multicolumn{1}{l}{\hspace{0.2cm}ok \scriptsize{\textcolor{darkgray}{ ($286\text{ms}$)}}}\hspace{-0.4cm} & \multicolumn{1}{l}{\hspace{0.2cm}ok \scriptsize{\textcolor{darkgray}{ ($241\text{ms}$)}}}\hspace{-0.4cm} & \multicolumn{1}{l}{\hspace{0.2cm}ok \scriptsize{\textcolor{darkgray}{ ($272\text{ms}$)}}}\hspace{-0.4cm} & \multicolumn{1}{l}{\hspace{0.2cm}ok \scriptsize{\textcolor{darkgray}{ ($265\text{ms}$)}}}\hspace{-0.4cm}\\
 & \small{MSVC} &\multicolumn{1}{l}{\hspace{0.2cm}ok \scriptsize{\textcolor{darkgray}{ ($269\text{ms}$)}}}\hspace{-0.4cm} &  - & \multicolumn{1}{l}{\hspace{0.2cm}ok \scriptsize{\textcolor{darkgray}{ ($322\text{ms}$)}}}\hspace{-0.4cm} & - \\
\multirow{3}{*}{\shortstack{\textbf{AES}\\\footnotesize{1 round}\\\footnotesize{\hspace{-0.15cm}85-110 vertices\hspace{-0.15cm}}}} & \small{GCC} &\multicolumn{1}{l}{\hspace{0.2cm}ok \scriptsize{\textcolor{darkgray}{ ($64\text{ms}$)}}}\hspace{-0.4cm} & \multicolumn{1}{l}{\hspace{0.2cm}ok \scriptsize{\textcolor{darkgray}{ ($61\text{ms}$)}}}\hspace{-0.4cm} & \multicolumn{1}{l}{\hspace{0.2cm}ok \scriptsize{\textcolor{darkgray}{ ($53\text{ms}$)}}}\hspace{-0.4cm} & \multicolumn{1}{l}{\hspace{0.2cm}ok \scriptsize{\textcolor{darkgray}{ ($56\text{ms}$)}}}\hspace{-0.4cm}\\
 & \small{Clang} &\multicolumn{1}{l}{\hspace{0.2cm}ok \scriptsize{\textcolor{darkgray}{ ($37\text{ms}$)}}}\hspace{-0.4cm} & \multicolumn{1}{l}{\hspace{0.2cm}ok \scriptsize{\textcolor{darkgray}{ ($32\text{ms}$)}}}\hspace{-0.4cm} & \multicolumn{1}{l}{\hspace{0.2cm}ok \scriptsize{\textcolor{darkgray}{ ($32\text{ms}$)}}}\hspace{-0.4cm} & \multicolumn{1}{l}{\hspace{0.2cm}ok \scriptsize{\textcolor{darkgray}{ ($27\text{ms}$)}}}\hspace{-0.4cm}\\
 & \small{MSVC} &\multicolumn{1}{l}{\hspace{0.2cm}ok \scriptsize{\textcolor{darkgray}{ ($30\text{ms}$)}}}\hspace{-0.4cm} &  - & \multicolumn{1}{l}{\hspace{0.2cm}ok \scriptsize{\textcolor{darkgray}{ ($42\text{ms}$)}}}\hspace{-0.4cm} & - \\
\end{tabular}
}

\caption{Signature matching step execution times, sample set of Lestringant et al.}
\vspace{-0.3cm}
 \label{tab:lestringanttable}
\end{threeparttable}
\end{table}

\subsection{Performance on OpenWRT binaries}

\begin{figure*}[!t]
\hspace*{-0.4cm}
 \begin{subfigure}{0.325\textwidth}
  \includegraphics[trim=0 0 0 0,clip,width=\textwidth]{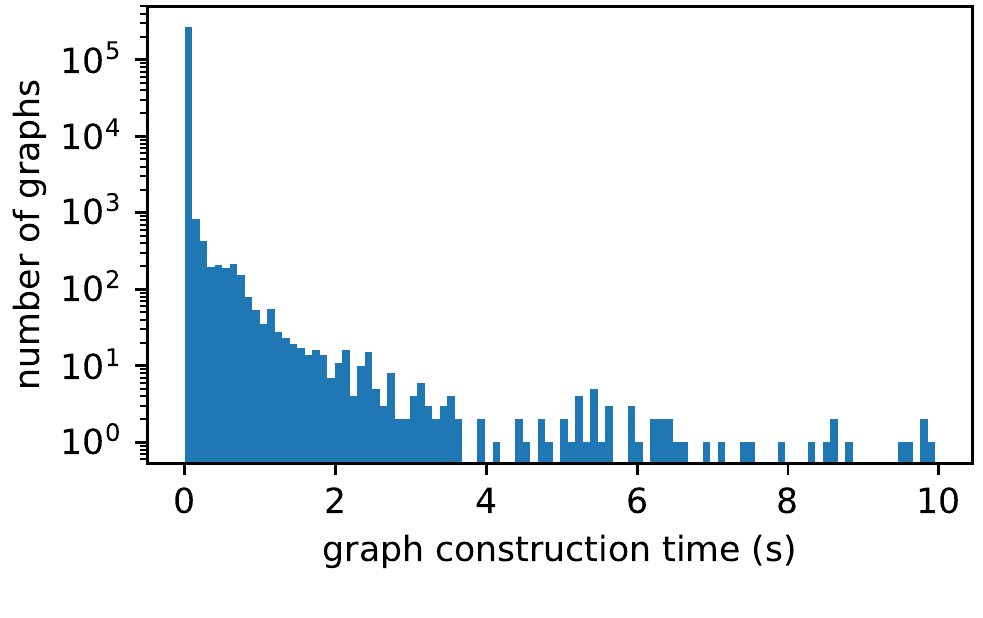}
  \caption{Histogram of graph construction}
  \label{fig:histogram} 
 \end{subfigure}
 \begin{subfigure}{0.251\textwidth}
  \includegraphics[trim=0 0 0 0,clip,width=\textwidth]{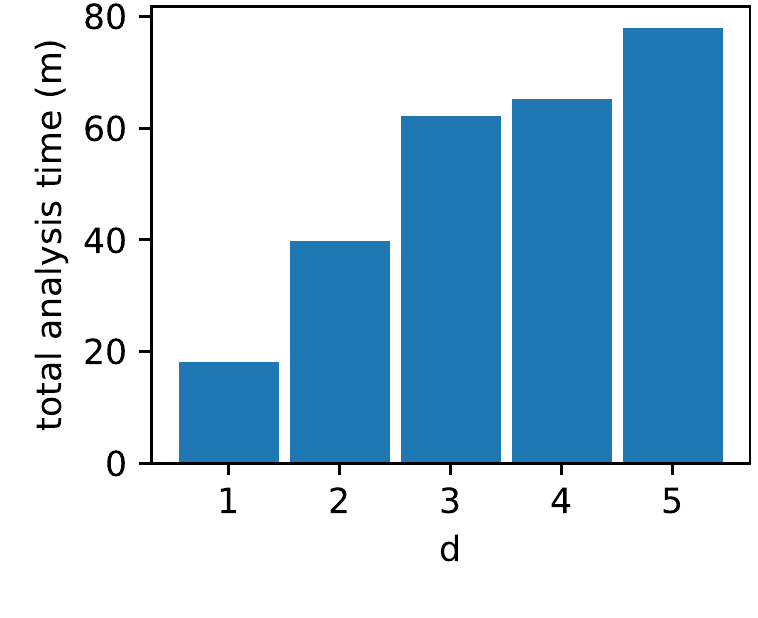}
  \caption{Inline depth $d$ vs analysis time}
  \label{fig:analysistime}
 \end{subfigure}
 \begin{subfigure}{0.435\textwidth}
  \hspace*{0.2cm}
  \includegraphics[trim=0 0 0 0,clip,width=\textwidth]{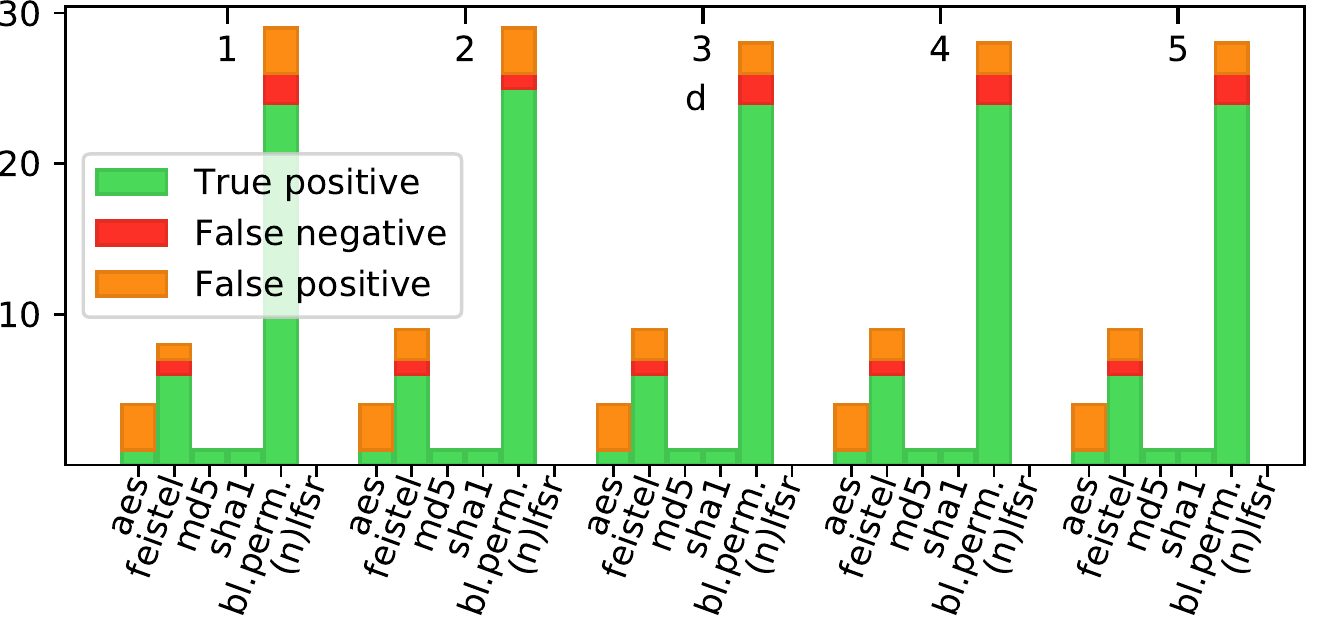}
  \caption{Inline depth $d$ vs accuracy}
  \label{fig:accuracy}
 \end{subfigure}
 \caption{Effect of inline depth $d$ and $t_\text{timeout}$ for libcrypto.so.1.1}
\vspace{-0.33cm}
\end{figure*}



The version of OpenWRT used is 19.07.2, which is the latest at time of writing. The sample set consists of several binaries taken from the distribution and is known to contain cryptographic primitives. 

DFG construction from binary code (Section~\ref{sec:dfgconstruct}) is a special case of execution, and is thus affected by the \emph{halting} problem. As such, graph-construction is not guaranteed to terminate. Therefore, we introduce a graph construction timeout $t_{\text{timeout}}$. Figure~\ref{fig:histogram} depicts a histogram of graph construction time $t$ for all graphs constructed during the analysis of libcrypto.so.1.1. It shows that, for the vast majority of all graphs, construction completes within $10s$. Thus, we take $t_{\text{timeout}} = 10s$.

Furthermore, we must decide what action to take when the function under analysis invokes another function. Either we perform inlining, and hence incorporate the entire invocation in the resulting DFG, or we represent it by a single \texttt{CALL} operation. To address this issue, we define a tunable variable $d$, denoting the depth level to which function calls are inlined. We investigate the impact of $d$ by running the analysis on \emph{libcrypto.so.1.1}, while taking on different values, and measuring performance in terms of running time and accuracy. We then choose a sensible value based on a trade-off between the two, and use it for the remainder of this section. Figure~\ref{fig:analysistime} depicts the time taken to complete the entire analysis pipeline over every function in libcrypto.so.1.1, under the influence of $d$. Figure~\ref{fig:accuracy} contains accuracy measurements for each signature. True negatives are omitted since they cover an overwhelming majority of results, and thus impact readability.

Recall that the signature evaluation is performed on graphs, and the graph construction step may yield several graphs. As such, several signature evaluation results may exist per function. The measurements provided in Figure~\ref{fig:accuracy} are aggregated on a per-function level.

Let $\fu$ be any function in the binary under analysis, and let signature $s_\alpha$ denote a signature targeting primitive $\alpha$. Furthermore, let $F$ be the set of DFGs generated from $\fu$ during the graph construction phase. Finally, $\text{match}(s_\alpha, G)$ indicates that signature $s_\alpha$ was identified in graph $G$, $\text{imp}(\fu, \alpha)$ denotes that $\fu$ implements cryptographic primitive $\alpha$.

\begin{table}[!b]

\vspace{-0.33cm}
\begin{threeparttable}[b]
\scalebox{0.51}{\begin{tabular}{l c c c c}
\multicolumn{2}{l}{\textbf{Algorithm}}\\
signature & \multicolumn{1}{c}{dropbear} & \multicolumn{1}{c}{libcrypto.so.1.1} & \multicolumn{1}{c}{libmbedcrypto.so.2.16.3\tnote{1}} & \multicolumn{1}{c}{libnettle.so.7.0\tnote{2}}\\
\hline\\[-0.3cm]
size & \multicolumn{1}{c}{145 KB} & \multicolumn{1}{c}{1,735 KB} & \multicolumn{1}{c}{197 KB} & \multicolumn{1}{c}{237 KB}\\
\parbox{0.1cm}{analysis~time} & \multicolumn{1}{c}{6m44s} & \multicolumn{1}{c}{39m47s} & \multicolumn{1}{c}{6m56s} & \multicolumn{1}{c}{11m32s}\\
\hline\\[-0.3cm]
\multicolumn{2}{l}{\textbf{SHA1}}\\
sha1    & \cellcolor{green!25} \checkmark~Unlabeled\tnote{3}   & \cellcolor{green!25} \checkmark~SHA1\_Update  & \cellcolor{green!25} \checkmark~sha1\_update\_ret     & \cellcolor{green!25} \checkmark~sha1\_compress\\
bl.perm.  & \cellcolor{green!25} \checkmark~Unlabeled\tnote{3}  & \cellcolor{green!25} \checkmark~SHA1\_Update  & \cellcolor{green!25} \checkmark~sha1\_update\_ret    & \cellcolor{green!25} \checkmark~sha1\_update\tnote{4}\\
\multicolumn{2}{l}{\textbf{SHA256}}\\
bl.perm.  & \cellcolor{green!25} \checkmark~Unlabeled\tnote{3}        & \cellcolor{green!25} \checkmark~SHA256\_Update\tnote{5}      & \cellcolor{green!25} \checkmark~sha256\_update\_ret  & \cellcolor{green!25} \checkmark~sha256\_update\tnote{4,5}\\
\multicolumn{2}{l}{\textbf{AES}}\\
aes     & \cellcolor{green!25} \checkmark~Unlabeled\tnote{3}     & \cellcolor{green!25} \checkmark~AES\_encrypt  & \cellcolor{green!25} \checkmark~aes\_encrypt & \cellcolor{green!25} \checkmark~aes\_encrypt\_armv6\\
\multicolumn{2}{l}{\textbf{MD4}}\\
bl.perm.  & \cellcolor{gray!25} N/A       & \cellcolor{green!25} \checkmark~MD4\_Update   & \cellcolor{gray!25} N/A       & \cellcolor{green!25} \checkmark~md4\_update\tnote{4}\\
\multicolumn{2}{l}{\textbf{MD5}}\\
md5     & \cellcolor{gray!25} N/A       & \cellcolor{green!25} \checkmark~MD5\_Update   & \cellcolor{green!25} \checkmark~md5\_update\_ret     & \cellcolor{green!25} \checkmark~hmac\_md5\_update\\
bl.perm.  & \cellcolor{gray!25} N/A       & \cellcolor{green!25} \checkmark~MD5\_Update   & \cellcolor{green!25} \checkmark~md5\_update\_ret     & \cellcolor{green!25} \checkmark~hmac\_md5\_update\\
\multicolumn{2}{l}{\textbf{RIPEMD160}}\\
bl.perm.  & \cellcolor{gray!25} N/A       & \cellcolor{green!25} \checkmark~RIPEMD160\_Update     & \cellcolor{gray!25} N/A       & \cellcolor{green!25} \checkmark~hmac\_ripemd160\_update\\
\multicolumn{2}{l}{\textbf{SHA512}}\\
bl.perm.  & \cellcolor{gray!25} N/A       & \cellcolor{green!25} \checkmark~SHA512\_Update\tnote{5}     & \cellcolor{green!25} \checkmark~sha512\_process\tnote{5}    & \cellcolor{green!25} \checkmark~sha512\_update\tnote{5}\\
\multicolumn{2}{l}{\textbf{SM3}}\\
bl.perm.  & \cellcolor{gray!25} N/A       & \cellcolor{green!25} \checkmark~sm3\_block\_data\_order       & \cellcolor{gray!25} N/A       & \cellcolor{gray!25} N/A\\
\multicolumn{2}{l}{\textbf{BLOWFISH}}\\
feistel & \cellcolor{gray!25} N/A       & \cellcolor{green!25} \checkmark~BF\_encrypt   & \cellcolor{green!25} \checkmark~blowfish\_crypt\_ecb\tnote{4}      & \cellcolor{green!25} \checkmark~blowfish\_encrypt\\
\multicolumn{2}{l}{\textbf{CAMELLIA}}\\
feistel & \cellcolor{gray!25} N/A       & \cellcolor{green!25} \checkmark~Camellia\_EncryptBlock        & \cellcolor{gray!25} N/A       & \cellcolor{green!25} \checkmark~camellia\_crypt\\
\multicolumn{2}{l}{\textbf{CAST}}\\
feistel & \cellcolor{gray!25} N/A       & \cellcolor{green!25} \checkmark~CAST\_ecb\_encrypt    & \cellcolor{gray!25} N/A       & \cellcolor{green!25} \checkmark~cast128\_encrypt\\
\multicolumn{2}{l}{\textbf{DES}}\\
feistel & \cellcolor{gray!25} N/A       & \cellcolor{green!25} \checkmark~DES\_encrypt2 & \cellcolor{gray!25} N/A       & \cellcolor{green!25} \checkmark~des\_encrypt\\
\multicolumn{2}{l}{\textbf{RC2}}\\
feistel & \cellcolor{gray!25} N/A       & \cellcolor{red!25} \xmark~RC2\_encrypt        & \cellcolor{gray!25} N/A       & \cellcolor{gray!25} N/A\\
\multicolumn{2}{l}{\textbf{SEED}}\\
feistel & \cellcolor{gray!25} N/A       & \cellcolor{green!25} \checkmark~SEED\_encrypt & \cellcolor{gray!25} N/A       & \cellcolor{gray!25} N/A\\
\multicolumn{2}{l}{\textbf{SM4}}\\
feistel & \cellcolor{gray!25} N/A       & \cellcolor{green!25} \checkmark~SM4\_encrypt  & \cellcolor{gray!25} N/A       & \cellcolor{gray!25} N/A\\
\multicolumn{2}{l}{\textbf{GOST}}\\
feistel & \cellcolor{gray!25} N/A       & \cellcolor{gray!25} N/A       & \cellcolor{gray!25} N/A       & \cellcolor{green!25} \checkmark~gosthash94\_digest\\
\multicolumn{2}{l}{\textbf{MD2}}\\
bl.perm.  & \cellcolor{gray!25} N/A       & \cellcolor{gray!25} N/A       & \cellcolor{gray!25} N/A       & \cellcolor{green!25} \checkmark~md2\_update\\
\multicolumn{2}{l}{\textbf{TWOFISH}}\\
feistel & \cellcolor{gray!25} N/A       & \cellcolor{gray!25} N/A       & \cellcolor{gray!25} N/A       & \cellcolor{red!25} \xmark~twofish\_encrypt\\
\multicolumn{2}{l}{\textbf{SHA3}}\\
bl.perm. & \cellcolor{gray!25} N/A       & \cellcolor{green!25} \checkmark~SHA3\_absorb       & \cellcolor{gray!25} N/A       & \cellcolor{green!25} \checkmark~sha3\_update\tnote{4}
\end{tabular}
}

\begin{tablenotes}
\item[1] \footnotesize Symbols prefixed with \texttt{mbedtls\_}
\item[2] \footnotesize Symbols prefixed with \texttt{nettle\_}
\item[3] \footnotesize \parbox{0.92\columnwidth}{Misclassified by IDA as an integer array. Manual cast to function required.}
\item[4] \footnotesize Positive match for $d \geq 4$.
\item[5] \footnotesize Positive match for $t_{\text{timeout}} \geq 30s$.
\end{tablenotes}

\caption{Analysis result for various binaries in OpenWRT}
 \label{tab:analysis}
\end{threeparttable}

\end{table}

A result is marked as a \emph{true positive} if $\text{imp}(\fu, \alpha) \land \exists_G. G \in F \land \text{match}(s_\alpha, G)$, i.e. $\fu$ implements cryptographic primitive $\alpha$, and its signature is found in \emph{at least one} graph in $F$. Indeed, there is no guarantee that all DFGs in $F$ contain algorithm $\alpha$, and hence it is expected that the signature is not found in every graph in $F$. A result is marked as a \emph{false positive} if $\neg\text{imp}(\fu, \alpha) \land \exists_G. G \in F \land \text{match}(s_\alpha, G)$, i.e. $\fu$ does not implement primitive $\alpha$, yet its signature is found in at least one graph in $F$. A result is marked as a \emph{true negative} if $\neg\text{imp}(\fu, \alpha) \land \neg\exists_G. G \in F \land \text{match}(s_\alpha, G)$. A result is a \emph{false negative} if $\text{imp}(\fu, \alpha) \land \neg\exists_G. G \in F \land \text{match}(s_\alpha, G)$.



The results in Figure~\ref{fig:accuracy} show that accuracy does not substantially improve when choosing $d > 2$. However, doing so does impact the running time. We conclude that, for libcrypto.so.1.1, $d=2$ is a reasonable trade-off between accuracy and running time. As such, we continue to use $d=2$ for the remainder of this section, unless specified otherwise.

At this point, sensible values for $n$, $d$ and $t_{\text{timeout}}$ have been selected. We continue the evaluation by feeding the entire set of OpenWRT binaries to our analysis framework. The results are listed in Table~\ref{tab:analysis}. Each cell in the table depicts the symbol name in the corresponding binary of the first positive result, or, in case of a false negative, the symbol name where a positive result is expected.
%
The results indicate our solution is capable of successfully identifying the vast majority of cryptographic primitives present in various binaries in a timely manner. Should accuracy take precedence over performance, it is possible to tune the parameters to improve detection.

\subsubsection{Discussion of invalid results}\label{subsubsec:discussion}
Table~\ref{tab:analysis} and Figure~\ref{fig:accuracy} contain several false positives and false negatives. In order to gain insights in the limitations of our approach, we highlight those instances here. 

\vspace{-3mm}
\paragraph{False negatives}
RC2 uses a regular addition, i.e. with carry over, rather than XOR, whereas the Feistel signature highlighted in Section~\ref{sec:feistel} relies on the XOR operation being present. Therefore, RC2 is not identified as a Feistel cipher.

Furthermore, SHA512 is consistently among the false negatives for the sequential block permutation class of primitives. This is due to a DFG consisting of $n$ (i.e. 4) instances of SHA512 being required for successful identification. However, said DFG consists of over 1,000,000 vertices, and causes the construction phase to exceed $t_\text{timeout}$. Increasing this value successfully mitigates the issue. However, it also affects the total analysis time. The exact same issue applies to SHA3 with $d \geq 3$, causing the Keccak-F function to be inlined, and consequently the construction to exceed $t_\text{timeout}$.

Twofish is a Feistel cipher with a complex round function. The Feistel signature used throughout the analysis supports a round function consisting of up to 8 consecutive arithmetic/logical operations, whereas the complexity of the Twofish round function goes beyond that. Unfortunately, extending the signature beyond 8 consecutive operations severely impacts the running time of our implementation. 

\vspace{-3mm}
\paragraph{False positives}\label{sec:falsepositives}
The AES key schedule is identified as a Feistel network. This is due to the fact that its structure can actually be formulated as one, i.e. each round $L_{i+1} = R_i$, and $R_{i+1} = L_i \oplus F(R_i, K_i)$, where $i$ denotes the round number for some function $F$. This is a perfect example to illustrate that the taxonomical tree of cryptographic primitives is not necessarily clear-cut. Rather, a degree of `fuzziness` exists among different classes.

RC4 and ChaCha, both stream ciphers, are identified as sequential block permutations. Inspection reveals that both implementations keep an internal state of some size $b$. The state is used directly as the cipher's keystream. After the internal state is fully consumed, a new internal state is generated. As such, the structure can be viewed as a special case of a block cipher with a block size of $b$ bytes.

\begin{table*}[!t]
\begin{threeparttable}[b]
\scalebox{0.72}{\begin{tabular}{l c l c c c l}
\textbf{Algorithm} & Type & Description & Reverse- & Cryptanalysis & Original & Target signature\\
			   & &			& engineered &			   &	source	   &				\\
\hline\\[-0.3cm]

\textbf{CRYPTO1}      &
\footnotesize{Stream} &
Cipher used in the Mifare Classic family of RFID tags. &
\cite{nohl2008reverse, garcia2008dismantling} &
\cite{de2008practical,garcia2008dismantling,garcia2009wirelessly,courtois2009dark,meijer2015ciphertext} &
\tablefootnote{\tiny\url{https://github.com/nfc-tools/mfcuk/blob/master/src/crypto1.c}} 
&
\cellcolor{green!25} \checkmark~(N)LFSR\tnote{1}\\

\textbf{HITAG2}       &
\footnotesize{Stream} &
Cipher used in vehicle immobilizers. &
\cite{wiener2007hitag2} &
\cite{courtois2009practical,soos2010enhanced,vstembera2011breaking,sun2011cube,verdult2012gone} &
\tablefootnote{\tiny\url{http://cryptolib.com/ciphers/hitag2/}}  &
\cellcolor{green!25} \checkmark~(N)LFSR\tnote{1}\\

\textbf{A5-1}   &
\footnotesize{Stream} &
Provides of over-the-air privacy for communication in GSM. &
\cite{briceno1999pedagogical} &
\cite{biham2000cryptanalysis,biryukov2000real,maximov2004improved} &
\tablefootnote{\tiny\url{https://cryptome.org/gsm-a512.htm}\label{ftn:a5}} &
\cellcolor{green!25} \checkmark~(N)LFSR\tnote{1}\\

\textbf{A5-2}   &
\footnotesize{Stream} &
GSM export cipher. &
\cite{briceno1999pedagogical} &
\cite{goldberg1999real} &
\footref{ftn:a5} &
\cellcolor{green!25} \checkmark~(N)LFSR\tnote{1}\\

\textbf{A5-GMR} &
\footnotesize{Stream} &
Cipher used in GMR, a standard for satellite phones. Heavily inspired by A5/2. &
\cite{driessen2012don} &
\cite{driessen2012don,driessen2013experimental} &
\tablefootnote{\tiny\url{https://github.com/marcelmaatkamp/gnuradio-osmocom-gmr/blob/master/src/l1/a5.c}} &
\cellcolor{green!25} \checkmark~(N)LFSR\tnote{1}\\

\textbf{RED PIKE}       &
\footnotesize{Block} &
Classified UK government encryption algorithm. &
\cite{cypherpunkredpike} &
- &
\tablefootnote{\tiny\url{https://en.wikipedia.org/wiki/Red_Pike_(cipher)}} &
\cellcolor{red!25} \xmark~Feistel cipher\\

\textbf{COMP128}        &
\footnotesize{Hash} &
Family of algorithms used for session key and MAC generation in GSM. &
\cite{bricenoimplementation,tamassecrets} &
\cite{brumley2004a3} &
\tablefootnote{\tiny\url{https://github.com/osmocom/libosmocore/blob/master/src/gsm/comp128.c}} &
\cellcolor{green!25} \checkmark~Block permutation\\

\textbf{KASUMI} &
\footnotesize{Block} &
Feistel cipher used for the confidentiality and integrity of 3G. &
- &
\cite{biham2005related,kim2010related,dunkelman2010practical} &
\tablefootnote{\tiny\url{https://github.com/osmocom/libosmocore/blob/master/src/gsm/kasumi.c}} &
\cellcolor{green!25} \checkmark~Feistel cipher\\

\textbf{MULTI2}        &
\footnotesize{Block} &
A block cipher used for broadcast scrambling in Japan. &
- &
\cite{aumasson2009cryptanalysis} &
\tablefootnote{\tiny\url{https://github.com/OP-TEE/optee_os/blob/master/core/lib/libtomcrypt/src/ciphers/multi2.c}} &
\cellcolor{green!25} \checkmark~Feistel cipher\\

\textbf{DST40}    &
\footnotesize{Block} &
Digital Signature Transponder cipher, often found in vehicle immobilizers. &
\cite{bono2005security} &
\cite{bono2005security} &
\tablefootnote{\tiny\url{https://github.com/jok40/dst40/blob/HEAD/software/dst40test/dst40.c}} &
\cellcolor{green!25} \checkmark~(N)LFSR\\

\textbf{KEELOQ}       &
\footnotesize{Block} &
Block cipher used in remote keyless entry systems and home automation. &
\cite{microchip1998keeloq} &
\cite{bogdanov2007cryptanalysis,courtois2008algebraic,biham2008steal,eisenbarth2008power} &
\tablefootnote{\tiny\url{https://github.com/hadipourh/KeeLoq}}\textsuperscript{,}\tablefootnote{\tiny\url{http://cryptolib.com/ciphers/keeloq/}} &
\cellcolor{green!25} \checkmark~(N)LFSR\\
\end{tabular}
}

\begin{tablenotes}
\item[1] \footnotesize Positive match for $d \geq 4$
\end{tablenotes}

\caption {Analysis result for proprietary samples}
\label{tab:proptable}

\end{threeparttable}
\vspace{-0.33cm}
\end{table*}

\begin{table}[htb]
\begin{threeparttable}[b]
\centering
\scalebox{0.70}{\begin{tabular}{l c c c c}
\multicolumn{2}{l}{\textbf{Algorithm}}\\
signature & \multicolumn{1}{c}{CWM0576} & \multicolumn{1}{c}{CWX0470} & \multicolumn{1}{c}{M340} & \multicolumn{1}{c}{VW}\\
\hline\\[-0.3cm]
size & \multicolumn{1}{c}{1,717 KB} & \multicolumn{1}{c}{1,344 KB} & \multicolumn{1}{c}{4,133 KB} & \multicolumn{1}{c}{512 KB}\\
\parbox{0.1cm}{analysis~time} & \multicolumn{1}{c}{88m14s} & \multicolumn{1}{c}{45m53s} & \multicolumn{1}{c}{83m11s} & \multicolumn{1}{c}{11m45s}\\
\hline\\[-0.3cm]
\multicolumn{2}{l}{\textbf{DES}}\\
feistel & \cellcolor{green!25} \checkmark~Match  & \cellcolor{green!25} \checkmark~Match & \cellcolor{gray!25} N/A       & \cellcolor{gray!25} N/A\\
\multicolumn{2}{l}{\textbf{AES}}\\
aes     & \cellcolor{green!25} \checkmark~Match & \cellcolor{gray!25} N/A       & \cellcolor{gray!25} N/A       & \cellcolor{gray!25} N/A\\
bl.perm.  & \cellcolor{green!25} \checkmark~Match        & \cellcolor{gray!25} N/A       & \cellcolor{gray!25} N/A       & \cellcolor{gray!25} N/A\\
\multicolumn{2}{l}{\textbf{MD5}}\\
md5     & \cellcolor{green!25} \checkmark~Match       & \cellcolor{green!25} \checkmark~Match       & \cellcolor{green!25} \checkmark~Match & \cellcolor{gray!25} N/A\\
bl.perm.  & \cellcolor{green!25} \checkmark~Match & \cellcolor{green!25} \checkmark~Match & \cellcolor{green!25} \checkmark~Match   & \cellcolor{gray!25} N/A\\
\multicolumn{2}{l}{\textbf{MEGAMOS}}\\
(n)lfsr & \cellcolor{gray!25} N/A       & \cellcolor{gray!25} N/A       & \cellcolor{gray!25} N/A       & \cellcolor{red!25} \xmark~No match
\end{tabular}
}

\caption{Analysis result for various firmware images}
 \label{tab:fwtable}
\end{threeparttable}
\vspace{-0.33cm}
\end{table}

Finally, CAST, ARIA and SM4 are all misidentified as AES. This is due to the fact that for all three primitives, either the algorithm itself, or its key schedule, is implemented by means of lookup tables in a fashion similar to that of AES. Ultimately, the transform completely depends on these tables, rather than information flows.


\subsection{Performance on proprietary algorithms}
Next, we turn our attention to various proprietary algorithms. Most algorithms were originally confidential, but have been leaked to the public or reverse engineered. As such, source code for all samples is publicly available. Due to uncertainty over the legality of redistribution, we point to the original sources for reference.
Table~\ref{tab:proptable} depicts the analysis results these algorithms. A description, the analysis result, and other relevant information is condensed into a single table due to space restrictions. All signatures target a generic class of primitives and none were pre-constructed to fit a particular sample. All algorithms are successfully identified, with the exception of Red Pike. Similar to RC2 from Section~\ref{subsubsec:discussion}, Red Pike uses addition instead of exclusive-or, and therefore not identified as a Feistel cipher.

Finally, the test set of representative real-world firmwares consists of images for the Emerson ControlWave Micro RTU, Emerson ControlWave XFC flow computer, Schneider Electric M340 PLC and Volkswagen IPC. The size, nature and complexity of these images ensure test-set realism.
Table~\ref{tab:fwtable} depicts the analysis result for all the firmwares. To the best of our knowledge, the table covers all cryptographic algorithms present in the sample set of firmware images. The images are `flat' binaries and hence symbol names are absent. The results show that all the cryptographic primitives were identified, except for the Megamos cipher. Verdult et al.~\cite{verdult2015dismantling} revealed that the Megamos cipher contains an NLFSR, and thus, the analysis should point this out. Further examination reveals that the non-linear feedback function is implemented as a subroutine, and the shift register is updated depending on its return value via an if-statement. 
This is a direct violation of the implicit flow limitation inherent to DFG-based approaches discussed in Section~\ref{sec:scope}.

\section{Conclusions}

Despite the ubiquitous availability of royalty-free, publicly documented, and peer-reviewed cryptographic primitives and implementations, proprietary alternatives have persisted across many industry verticals, especially in embedded systems.
Due to the undocumented and proprietary nature of said primitives, subjecting them to security analysis often requires locating and classifying them in often very large binary images, which is a time-consuming, labor-intensive effort. 

In order to overcome this obstacle in an automated fashion, a solution should have the capability of identifying as-of-yet unknown cryptographic algorithms, support large, real-world firmware binaries, and not depend on peripheral emulation. As of yet, no prior work exists that satisfies these criteria.

Our novel approach combines DFG isomorphism with symbolic execution, and introduces a specialized DSL in order to enable identification of unknown proprietary cryptographic algorithms falling within well-defined taxonomical classes. The approach is the first of its kind, is architecture and platform agnostic, and performs well in terms of both accuracy and running time on real-world binary firmware images.
\vspace{-3mm}
\paragraph{Future work}
DFGs do not allow for the expression of code flow information. Potentially valuable indicators, such as whether two nodes originate from the same execution address, hinting to a round function, are therefore lost. We leave the incorporation of code flow information for future work.


\bibliographystyle{plain}
\bibliography{bibliography}

\begin{thebibliography}{10}

\bibitem{anderson2006}
Ross Anderson, Mike Bond, Jolyon Clulow, and Sergei Skorobogatov.
\newblock Cryptographic processors-a survey.
\newblock {\em Proceedings of the IEEE}, 94(2):357--369, 2006.

\bibitem{aumasson2009cryptanalysis}
Jean-Philippe Aumasson, Jorge Nakahara, and Pouyan Sepehrdad.
\newblock Cryptanalysis of the isdb scrambling algorithm (multi2).
\newblock In {\em International Workshop on Fast Software Encryption}, pages
  296--307. Springer, 2009.

\bibitem{auriemma2013signsrch}
Luigi Auriemma.
\newblock Signsrch tool.
\newblock {\em tool for searching signatures inside files}, 2013.

\bibitem{avanzi2016salad}
Roberto Avanzi.
\newblock A salad of block ciphers.
\newblock {\em IACR Cryptology ePrint Archive}, 2016:1171, 2016.

\bibitem{bbc2013car}
{BBC News}.
\newblock Car key immobiliser hack revelations blocked by uk court.
\newblock 2013.
\newblock \url{https://www.bbc.com/news/technology-23487928}.

\bibitem{biham2000cryptanalysis}
Eli Biham and Orr Dunkelman.
\newblock Cryptanalysis of the a5/1 gsm stream cipher.
\newblock In {\em International Conference on Cryptology in India}, pages
  43--51. Springer, 2000.

\bibitem{biham2008steal}
Eli Biham, Orr Dunkelman, Sebastiaan Indesteege, Nathan Keller, and Bart
  Preneel.
\newblock How to steal cars a practical attack on keeloq.
\newblock In {\em EUROCRYPT}, pages 1--18, 2008.

\bibitem{biham2005related}
Eli Biham, Orr Dunkelman, and Nathan Keller.
\newblock A related-key rectangle attack on the full kasumi.
\newblock In {\em International Conference on the Theory and Application of
  Cryptology and Information Security}, pages 443--461. Springer, 2005.

\bibitem{biondi2017effectiveness}
Fabrizio Biondi, S{\'e}bastien Josse, Axel Legay, and Thomas Sirvent.
\newblock Effectiveness of synthesis in concolic deobfuscation.
\newblock {\em Computers \& Security}, 70:500--515, 2017.

\bibitem{biryukov2000real}
Alex Biryukov, Adi Shamir, and David Wagner.
\newblock Real time cryptanalysis of a5/1 on a pc.
\newblock In {\em International Workshop on Fast Software Encryption}, pages
  1--18. Springer, 2000.

\bibitem{blazytko2017syntia}
Tim Blazytko, Moritz Contag, Cornelius Aschermann, and Thorsten Holz.
\newblock Syntia: Synthesizing the semantics of obfuscated code.
\newblock In {\em Proceedings of the 26th {USENIX} Security Symposium}, pages
  643--659, 2017.

\bibitem{bogdanov2007cryptanalysis}
Andrey Bogdanov.
\newblock Cryptanalysis of the keeloq block cipher.
\newblock {\em IACR Cryptology ePrint Archive}, 2007:55, 2007.

\bibitem{bokslagassessment}
Wouter Bokslag.
\newblock An assessment of ecm authentication in modern vehicles.

\bibitem{bono2005security}
Steve Bono, Matthew Green, Adam Stubblefield, Ari Juels, Aviel~D Rubin, and
  Michael Szydlo.
\newblock Security analysis of a cryptographically-enabled rfid device.
\newblock In {\em USENIX Security Symposium}, volume~31, pages 1--16, 2005.

\bibitem{bricenoimplementation}
Marc Briceno, Ian Goldberg, and David Wagner.
\newblock An implementation of comp128.
\newblock 1998.
\newblock \url{http://www.iol.ie/kooltek/a3a8.txt}.

\bibitem{briceno1999pedagogical}
Marc Briceno, Ian Goldberg, and David Wagner.
\newblock A pedagogical implementation of the gsm a5/1 and a5/2 “voice
  privacy” encryption algorithms.
\newblock {\em Originally published at http://www. scard. org, mirror at
  http://cryptome. org/gsm-a512. htm}, 26, 1999.

\bibitem{brumley2004a3}
Billy Brumley.
\newblock A3/a8 \& comp128.
\newblock {\em T-79.514 Special Course on Cryptology}, pages 1--18, 2004.

\bibitem{caballero2009dispatcher}
Juan Caballero, Pongsin Poosankam, Christian Kreibich, and Dawn Song.
\newblock Dispatcher: Enabling active botnet infiltration using automatic
  protocol reverse-engineering.
\newblock In {\em Proceedings of the 16th ACM conference on Computer and
  communications security}, pages 621--634, 2009.

\bibitem{calvet2012aligot}
Joan Calvet, Jos{\'e}~M Fernandez, and Jean-Yves Marion.
\newblock Aligot: cryptographic function identification in obfuscated binary
  programs.
\newblock In {\em Proceedings of the 2012 ACM conference on Computer and
  communications security}, pages 169--182, 2012.

\bibitem{courtois2009dark}
Nicolas~T Courtois.
\newblock The dark side of security by obscurity and cloning mifare classic
  rail and building passes, anywhere, anytime.
\newblock 2009.

\bibitem{courtois2008algebraic}
Nicolas~T Courtois, Gregory~V Bard, and David Wagner.
\newblock Algebraic and slide attacks on keeloq.
\newblock In {\em International Workshop on Fast Software Encryption}, pages
  97--115. Springer, 2008.

\bibitem{courtois2009practical}
Nicolas~T Courtois, Sean O’Neil, and Jean-Jacques Quisquater.
\newblock Practical algebraic attacks on the hitag2 stream cipher.
\newblock In {\em International Conference on Information Security}, pages
  167--176. Springer, 2009.

\bibitem{cypherpunkredpike}
Gmane Cypherpunk~mailing list.
\newblock Red pike cipher.
\newblock 2004.
\newblock
  \url{http://permalink.gmane.org/gmane.comp.security.cypherpunks/3680}.

\bibitem{david2017formal}
Robin David.
\newblock {\em Formal Approaches for Automatic Deobfuscation and
  Reverse-engineering of Protected Codes}.
\newblock PhD thesis, 2017.

\bibitem{de2008practical}
Gerhard de~Koning~Gans, Jaap-Henk Hoepman, and Flavio~D Garcia.
\newblock A practical attack on the mifare classic.
\newblock In {\em International Conference on Smart Card Research and Advanced
  Applications}, pages 267--282. Springer, 2008.

\bibitem{driessen2012don}
Benedikt Driessen, Ralf Hund, Carsten Willems, Christof Paar, and Thorsten
  Holz.
\newblock Don't trust satellite phones: A security analysis of two satphone
  standards.
\newblock In {\em 2012 IEEE Symposium on Security and Privacy}, pages 128--142.
  IEEE, 2012.

\bibitem{driessen2013experimental}
Benedikt Driessen, Ralf Hund, Carsten Willems, Christof Paar, and Thorsten
  Holz.
\newblock An experimental security analysis of two satphone standards.
\newblock {\em ACM Transactions on Information and System Security (TISSEC)},
  16(3):1--30, 2013.

\bibitem{dunkelman2010practical}
Orr Dunkelman, Nathan Keller, and Adi Shamir.
\newblock A practical-time related-key attack on the kasumi cryptosystem used
  in gsm and 3g telephony.
\newblock In {\em Annual cryptology conference}, pages 393--410. Springer,
  2010.

\bibitem{eisenbarth2008power}
Thomas Eisenbarth, Timo Kasper, Amir Moradi, Christof Paar, Mahmoud
  Salmasizadeh, and Mohammad T~Manzuri Shalmani.
\newblock On the power of power analysis in the real world: A complete break of
  the keeloq code hopping scheme.
\newblock In {\em Annual International Cryptology Conference}, pages 203--220.
  Springer, 2008.

\bibitem{etsi300}
ETSI.
\newblock 300 392-7 v3. 3.1 (2012-07) european standard (telecommunication
  series) terrestrial trunked radio (tetra); voice plus data (v+ d); part 7:
  Security.
\newblock {\em European Telecommunications Standards Institute (ETSI)}, 2012.
\newblock
  \url{https://www.etsi.org/deliver/etsi_en/300300_300399/30039207/03.03.01_60/en_30039207v030301p.pdf}.

\bibitem{garba2019saturn}
Peter Garba and Matteo Favaro.
\newblock Saturn-software deobfuscation framework based on llvm.
\newblock In {\em Proceedings of the 3rd ACM Workshop on Software Protection},
  pages 27--38, 2019.

\bibitem{garcia2008dismantling}
Flavio~D Garcia, Gerhard de~Koning~Gans, Ruben Muijrers, Peter Van~Rossum, Roel
  Verdult, Ronny~Wichers Schreur, and Bart Jacobs.
\newblock Dismantling mifare classic.
\newblock In {\em European symposium on research in computer security}, pages
  97--114. Springer, 2008.

\bibitem{garcia2009wirelessly}
Flavio~D Garcia, Peter Van~Rossum, Roel Verdult, and Ronny~Wichers Schreur.
\newblock Wirelessly pickpocketing a mifare classic card.
\newblock In {\em 2009 30th IEEE Symposium on Security and Privacy}, pages
  3--15. IEEE, 2009.

\bibitem{goldberg1999real}
Ian Goldberg, David Wagner, and Lucky Green.
\newblock The real-time cryptanalysis of a5/2.
\newblock {\em Rump session of Crypto}, 99:16, 1999.

\bibitem{grobert2011}
Felix Gr{\"o}bert, Carsten Willems, and Thorsten Holz.
\newblock Automated identification of cryptographic primitives in binary
  programs.
\newblock In {\em Recent Advances in Intrusion Detection}, pages 41--60, 2011.

\bibitem{guilfanovfindcrypt}
Ilfak Guilfanov.
\newblock Findcrypt2, february 2006.
\newblock \url{http://www.hexblog.com/?p=28}.

\bibitem{gutmann2003}
Peter Gutmann.
\newblock {\em Cryptographic security architecture: design and verification}.
\newblock Springer Science \& Business Media, 2003.
\newblock pages 293.

\bibitem{hill2017deep}
Gregory~D Hill and Xavier~JA Bellekens.
\newblock Deep learning based cryptographic primitive classification.
\newblock {\em arXiv preprint arXiv:1709.08385}, 2017.

\bibitem{keliher2003linear}
Liam~Timothy Keliher.
\newblock {\em Linear cryptanalysis of substitution-permutation networks}.
\newblock Queen's University, 2003.

\bibitem{ker_1883_militaire}
Auguste Kerckhoffs.
\newblock La cryptographie militaire.
\newblock {\em Journal des Sciences Militaires}, IX:5--83, 161--191, 1883.

\bibitem{kim2010related}
Jongsung Kim, Seokhie Hong, Bart Preneel, Eli Biham, Orr Dunkelman, and Nathan
  Keller.
\newblock Related-key boomerang and rectangle attacks.
\newblock {\em IACR Cryptology ePrint Archive}, 2010:19, 2010.

\bibitem{lagadec2014balbuzard}
Philippe Lagadec.
\newblock Balbuzard, 2014.
\newblock \url{http://www.decalage.info/en/python/balbuzard}.

\bibitem{lestringant2015automated}
Pierre Lestringant, Fr{\'e}d{\'e}ric Guih{\'e}ry, and Pierre-Alain Fouque.
\newblock Automated identification of cryptographic primitives in binary code
  with data flow graph isomorphism.
\newblock In {\em Proceedings of the 10th ACM Symposium on Information,
  Computer and Communications Security}, pages 203--214. ACM, 2015.

\bibitem{levin2013draft}
Literatecode.
\newblock Draft crypto analyzer (draca).
\newblock \url{http://www.literatecode.com/draca}, May 2013.

\bibitem{loki2008snd}
Loki.
\newblock Snd crypto scanner (olly/immunity plugin), 2008.
\newblock
  \url{https://web.archive.org/web/20080321134709/http://tuts4you.com/forum/index.php?showtopic=15447}.

\bibitem{manifavas2016survey}
Charalampos Manifavas, George Hatzivasilis, Konstantinos Fysarakis, and Yannis
  Papaefstathiou.
\newblock A survey of lightweight stream ciphers for embedded systems.
\newblock {\em Security and Communication Networks}, 9(10):1226--1246, 2016.

\bibitem{matenaar2012cis}
Felix Matenaar, Andre Wichmann, Felix Leder, and Elmar Gerhards-Padilla.
\newblock Cis: The crypto intelligence system for automatic detection and
  localization of cryptographic functions in current malware.
\newblock In {\em 2012 7th International Conference on Malicious and Unwanted
  Software}, pages 46--53. IEEE, 2012.

\bibitem{maximov2004improved}
Alexander Maximov, Thomas Johansson, and Steve Babbage.
\newblock An improved correlation attack on a5/1.
\newblock In {\em International Workshop on Selected Areas in Cryptography},
  pages 1--18. Springer, 2004.

\bibitem{meijer2015ciphertext}
Carlo Meijer and Roel Verdult.
\newblock Ciphertext-only cryptanalysis on hardened mifare classic cards.
\newblock In {\em Proceedings of the 22nd ACM SIGSAC Conference on Computer and
  Communications Security}, pages 18--30, 2015.

\bibitem{menezes1996handbook}
Alfred~J Menezes, Jonathan Katz, Paul~C Van~Oorschot, and Scott~A Vanstone.
\newblock {\em Handbook of applied cryptography}.
\newblock CRC press, 1996.

\bibitem{microchip1998keeloq}
Microchip.
\newblock Hopping code decoder using a {PIC16C56}, {AN642}.
\newblock 1998.
\newblock
  \url{https://web.archive.org/web/20080916043223/http://www.keeloq.boom.ru/decryption.pdf}.

\bibitem{paradox2009hcd}
{Mr Paradox, AT4RE}.
\newblock Hash \& crypto detector (hcd), 2009.
\newblock
  \url{https://web.archive.org/web/20091203010936/http://www.at4re.com/download.php?view.8}.

\bibitem{newsome2005dynamic}
James Newsome and Dawn~Xiaodong Song.
\newblock Dynamic taint analysis for automatic detection, analysis, and
  signaturegeneration of exploits on commodity software.
\newblock In {\em NDSS}, volume~5, pages 3--4. Citeseer, 2005.

\bibitem{nohl2008reverse}
Karsten Nohl, David Evans, Starbug Starbug, and Henryk Pl{\"o}tz.
\newblock Reverse-engineering a cryptographic rfid tag.
\newblock In {\em USENIX security symposium}, volume~28, 2008.

\bibitem{nohl2010cryptanalysis}
Karsten Nohl, Erik Tews, and Ralf-Philipp Weinmann.
\newblock Cryptanalysis of the dect standard cipher.
\newblock In {\em International Workshop on Fast Software Encryption}, pages
  1--18. Springer, 2010.

\bibitem{plohmann2012simplifire}
Daniel Plohmann and Alexander Hanel.
\newblock simplifire. idascope, 2012.

\bibitem{salwan2018symbolic}
Jonathan Salwan, S{\'e}bastien Bardin, and Marie-Laure Potet.
\newblock Symbolic deobfuscation: From virtualized code back to the original.
\newblock In {\em International Conference on Detection of Intrusions and
  Malware, and Vulnerability Assessment}, pages 372--392. Springer, 2018.

\bibitem{snaker2015kanal}
{snaker, Maxx}.
\newblock Kanal - krypto analyzer for peid, 2015.
\newblock
  \url{http://www.dcs.fmph.uniba.sk/zri/6.prednaska/tools/PEiD/plugins/kanal.htm}.

\bibitem{soos2010enhanced}
Mate Soos.
\newblock Enhanced gaussian elimination in dpll-based sat solvers.
\newblock In {\em POS@ SAT}, pages 2--14, 2010.

\bibitem{vstembera2011breaking}
Petr {\v{S}}tembera and Martin Novotny.
\newblock Breaking hitag2 with reconfigurable hardware.
\newblock In {\em 2011 14th Euromicro Conference on Digital System Design},
  pages 558--563. IEEE, 2011.

\bibitem{strobel2013fuming}
Daehyun Strobel, Benedikt Driessen, Timo Kasper, Gregor Leander, David Oswald,
  Falk Schellenberg, and Christof Paar.
\newblock Fuming acid and cryptanalysis: Handy tools for overcoming a digital
  locking and access control system.
\newblock In {\em Annual Cryptology Conference}, pages 147--164. Springer,
  2013.

\bibitem{sun2011cube}
Siwei Sun, Lei Hu, Yonghong Xie, and Xiangyong Zeng.
\newblock Cube cryptanalysis of hitag2 stream cipher.
\newblock In {\em International Conference on Cryptology and Network Security},
  pages 15--25. Springer, 2011.

\bibitem{tamassecrets}
Jos Tamas.
\newblock Secrets of the sim.
\newblock 2013.
\newblock \url{http://www.hackingprojects.net/2013/04/secrets-of-sim.html}.

\bibitem{tofighi2019defeating}
Ramtine Tofighi-Shirazi, Irina-Mariuca Asavoae, Philippe Elbaz-Vincent, and
  Thanh-Ha Le.
\newblock Defeating opaque predicates statically through machine learning and
  binary analysis.
\newblock In {\em Proceedings of the 3rd ACM Workshop on Software Protection},
  pages 3--14, 2019.

\bibitem{tofighi2018dose}
Ramtine Tofighi-Shirazi, Maria Christofi, Philippe Elbaz-Vincent, and Thanh-Ha
  Le.
\newblock Dose: Deobfuscation based on semantic equivalence.
\newblock In {\em Proceedings of the 8th Software Security, Protection, and
  Reverse Engineering Workshop}, pages 1--12, 2018.

\bibitem{ullmann1976}
Julian~R Ullmann.
\newblock An algorithm for subgraph isomorphism.
\newblock {\em Journal of the ACM (JACM)}, 23(1):31--42, 1976.

\bibitem{verdult2015security}
Roel Verdult.
\newblock {\em The (in) security of proprietary cryptography}.
\newblock PhD thesis, [Sl: sn], 2015.

\bibitem{verdult2012gone}
Roel Verdult, Flavio~D Garcia, and Josep Balasch.
\newblock Gone in 360 seconds: Hijacking with hitag2.
\newblock In {\em Presented as part of the 21st {USENIX} Security Symposium},
  pages 237--252, 2012.

\bibitem{verdult2015dismantling}
Roel Verdult, Flavio~D Garcia, and Baris Ege.
\newblock Dismantling megamos crypto: Wirelessly lockpicking a vehicle
  immobilizer.
\newblock In {\em Supplement to the Proceedings of 22nd {USENIX} Security
  Symposium}, pages 703--718, 2015.

\bibitem{verstegen2018press}
Aram Verstegen, Peter Schwabe, Iskander Kuijer, and Roel Verdult.
\newblock Press to unlock: Analysis, reverse-engineering and implementation of
  hitag2-based remote keyless entry systems.
\newblock 2018.

\bibitem{weiner2013security}
Michael Weiner, Maurice Massar, Erik Tews, Dennis Giese, and Wolfgang Wieser.
\newblock Security analysis of a widely deployed locking system.
\newblock In {\em Proceedings of the 2013 ACM SIGSAC conference on Computer \&
  communications security}, pages 929--940, 2013.

\bibitem{wiener2007hitag2}
I.C. Wiener.
\newblock Hitag2 specification, reference implementation and test vectors,
  2007.
\newblock \url{http://cryptolib.com/ciphers/hitag2}.

\bibitem{wouters2019}
Lennert Wouters, Eduard Marin, Tomer Ashur, Benedikt Gierlichs, and Bart
  Preneel.
\newblock Fast, furious and insecure: Passive keyless entry and start systems
  in modern supercars.
\newblock {\em IACR Transactions on Cryptographic Hardware and Embedded
  Systems}, 2019(3):66--85, May 2019.

\bibitem{x3chun2004crypto}
x3chun.
\newblock Crypto searcher, 2004.
\newblock
  \url{https://web.archive.org/web/20050211180634/http://x3chun.wo.to/}.

\bibitem{xu2018vmhunt}
Dongpeng Xu, Jiang Ming, Yu~Fu, and Dinghao Wu.
\newblock Vmhunt: A verifiable approach to partially-virtualized binary code
  simplification.
\newblock In {\em Proceedings of the 2018 ACM SIGSAC Conference on Computer and
  Communications Security}, pages 442--458, 2018.

\bibitem{xu2017cryptographic}
Dongpeng Xu, Jiang Ming, and Dinghao Wu.
\newblock Cryptographic function detection in obfuscated binaries via
  bit-precise symbolic loop mapping.
\newblock In {\em 2017 IEEE Symposium on Security and Privacy (SP)}, pages
  921--937. IEEE, 2017.

\bibitem{yadegari2016automatic}
Babak Yadegari.
\newblock Automatic deobfuscation and reverse engineering of obfuscated code.
\newblock 2016.

\end{thebibliography}

\appendix

\section{Path Oracle Policy -- an example} 
\label{appendix:oraclePolicy}

\begin{figure}[h!tb]
\includegraphics[width=\columnwidth]{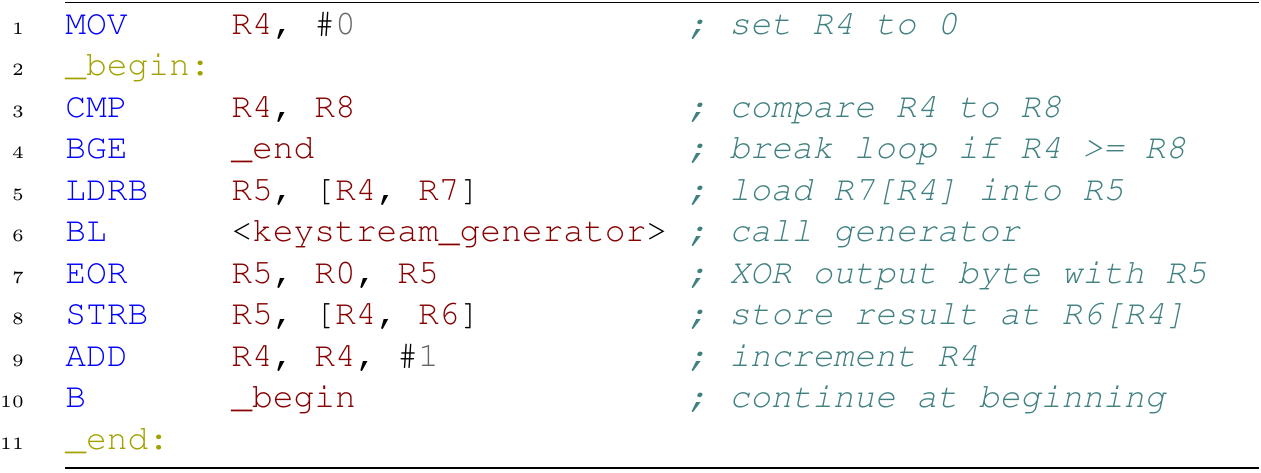}
\caption{Example stream cipher ARM assembly snippet}
\label{fig:snippet}
\end{figure}

\noindent Suppose the graph construction is run on the example ARM assembly snippet given in Figure~\ref{fig:snippet}. We start with $\St = (G, P, B)$, with $P = \textbf{true}$. Line $4$ contains conditional instruction \emph{Branch Greater/Equal (\texttt{BGE})}. During the first visit of this instruction, we have $i=0$, $P = \textbf{true}$, and $c = (\texttt{R8} \leq 0)$. Since the value of \texttt{R8} is unknown, $c$ is underdetermined. The path oracle policy prescribes \texttt{TAKE\_BOTH}. Thus, we get $P = (\texttt{R8} \leq 0)$, $B_4[0] = \textbf{true}$, and $\St' = (G', P', B')$, with $P' = (\texttt{R8} > 0)$ and $B'_4[0] = \textbf{false}$.
For state $\St$, the instruction is evaluated, and thus the construction continues on line $11$, and hence terminates. For $\St'$, the instruction is skipped, thereby visiting the body of the loop. Eventually, $\St'$ revisits the instruction at line $4$. This time we have $c = (\texttt{R8} \leq 1)$, $i=1$, $P' = (\texttt{R8} > 0)$ and $B'_4[0] = \textbf{false}$. Since $P' \land c$ is underdetermined, we query the path oracle, and obtain \texttt{TAKE\_FALSE}, causing another visit of the loop's body.
Finally, at $i=n$, we get $c = \texttt{R8} \leq n$ and $P' = (\texttt{R8} > n-1)$. We obtain \texttt{TAKE\_TRUE} from the path oracle. Thus, the construction terminates.
We obtain two graphs; one corresponding to predicate $\texttt{R8} \leq 0$, and another corresponding to $\texttt{R8} = n$. The latter describes $n$ iterations of the algorithm, exactly conforming to our goal. The former describes zero iterations, and thus, contains a negligible amount of nodes. Therefore, we accept the small amount of overhead this graph induces during later stages of the analysis.

\section{Miscellaneous rewrite rules}\label{appendix:misctrans}
Besides the rewrite rules already described, we apply additional miscellaneous rules. They
were conceived through continuous application of our framework to code fragments from various sources, and subsequent stumbling upon variations between the processed result generated from supposedly semantically equivalent code. We highlight these rules below. Different compilers have different optimization strategies. As such, some finetuning of these rules may be necessary when analyzing code produced by a vastly different compiler than those already accounted for.

There are various means of doubling the value of an arbitrary expression $x$. For example, $\texttt{MULT}(x, 2)$, but also $\texttt{ADD}(x,x)$ and $x\text{\textless{}\textless{}}1$. We represent all variants by $\texttt{MULT}(x, 2)$.

\begin{figure*}[!b]
\centering
 \begin{subfigure}[t]{0.45\textwidth}
 \centering
  \includegraphics[width=0.9\columnwidth]{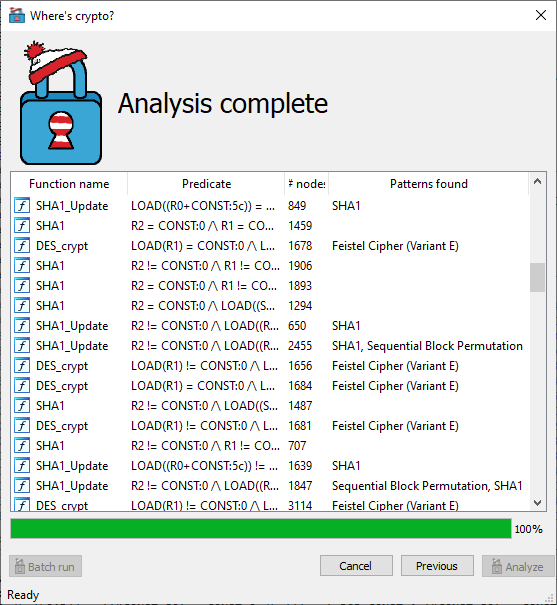}
  \caption{Sample analysis report}
 \end{subfigure}
 \begin{subfigure}[t]{0.5\textwidth}
 \centering
  \includegraphics[width=0.9\columnwidth]{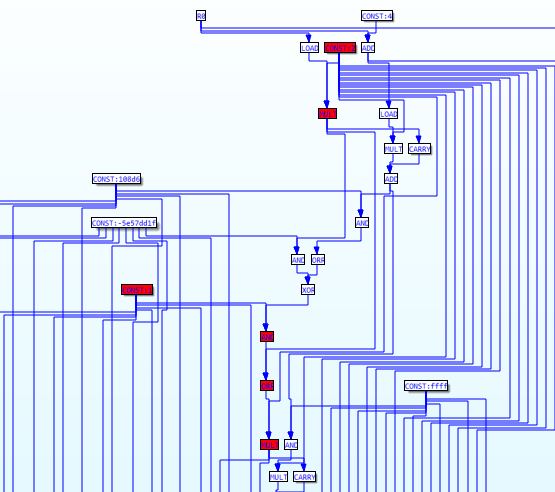}
  \caption{DFG plot generated from assembly, highlighting an LFSR}
 \end{subfigure}
\caption{An impression of the implementation of our framework}
\label{fig:poc} 
\end{figure*}

\vspace*{-0.25cm}\begin{center}
\scalebox{0.8}{
\begin{tikzpicture}[auto]
\node[data] (r1) {\footnotesize R1};
\path (r1)+(0.375, -0.8) node[op,circle,inner sep=0.1cm] (add) {+};

\path[->] (r1) edge[bend left=15] node[] {} (add);
\path[->] (r1) edge[bend right=15] node[] {} (add);

\path (add)+(1,0.4) node [] (arrow1) {};
\path (arrow1.center)+(1,0) node[] (arrow2) {};
\path[->,line width=3pt] (arrow1.center) edge [] node[] {} (arrow2.center);

\path (arrow2)+(1,0.4) node [data] (r1p) {\footnotesize R1};
\path (r1p)+(0.75,0) node [data] (2) {\footnotesize \texttt{2}};
\path (r1p)+(0.375,-0.8) node [data,circle,inner sep=0.05cm] (mult) {\scriptsize \textsc{mult}};

\path[->] (r1p) edge[] node[] {} (mult);
\path[->] (2) edge[] node[] {} (mult);
\end{tikzpicture}
}
\end{center}

\noindent Furthermore, suppose we have an arbitrary expression $x$, and constants $c_1$ and $c_2$. Then, the results of $\texttt{AND}(x\text{~\textgreater{}\textgreater{}~}c_1, c_2)$ and $\texttt{AND}(\texttt{ROTATE}(x,c_1),c_2)$, are equivalent if $c_2 < 2^{32 - c_1}$ and $c_1 < 32$, for a 32-bit architecture. This equivalence is sometimes exploited by compilers. 
In such a scenario, we represent both variants by $\texttt{AND}(x\text{~\textgreater{}\textgreater{}~}c_1, c_2)$.

\vspace*{-0.25cm}\begin{center}
\scalebox{0.8}{
\begin{tikzpicture}[auto]
\node[data] (r4) {\footnotesize R4};
\path (r4)+(0.375, -0.8) node[op,circle,inner sep=0.05cm] (rrx) {\footnotesize \textsc{rot}};
\path (r4)+(0.75,0) node[data] (8) {\footnotesize \texttt{8}};
\path (rrx)+(0.38,-0.8) node[op,circle,inner sep=0.05cm] (and) {\footnotesize \textsc{and}};
\path (rrx)+(0.8,0) node[data] (const) {\footnotesize \texttt{0xff}};

\path[->] (r4) edge[] node[] {} (rrx);
\path[->] (8) edge[] node[] {} (rrx);
\path[->] (rrx) edge[] node[] {} (and);
\path[->] (const) edge[] node[] {} (and);

\path (const)+(1,0) node [] (arrow1) {};
\path (arrow1.center)+(1,0) node[] (arrow2) {};
\path[->,line width=3pt] (arrow1.center) edge [] node[] {} (arrow2.center);

\path (arrow2)+(0.8,0.8) node[data] (r4p) {\footnotesize R4};
\path (r4p)+(0.375, -0.8) node[op,circle,inner sep=0.1cm] (rsh) {\footnotesize \textsc{\textgreater{}\textgreater{}}};
\path (r4p)+(0.75,0) node[data] (8p) {\footnotesize \texttt{8}};
\path (rsh)+(0.38,-0.8) node[op,circle,inner sep=0.05cm] (andp) {\footnotesize \textsc{and}};
\path (rsh)+(0.8,0) node[data] (constp) {\footnotesize \texttt{0xff}};

\path[->] (r4p) edge[] node[] {} (rsh);
\path[->] (8p) edge[] node[] {} (rsh);
\path[->] (rsh) edge[] node[] {} (andp);
\path[->] (constp) edge[] node[] {} (andp);
\end{tikzpicture}
}
\end{center}

\vspace*{-0.35cm}\noindent Lastly, we distribute multiplications over additions.

\vspace*{-0.25cm}\begin{center}
\scalebox{0.8}{
\begin{tikzpicture}[auto]
\node[data] (r3) {\footnotesize R3};
\path (r3)+(0.375, -0.8) node[op,circle,inner sep=0.1cm] (add1) {+};
\path (r3)+(0.75,0) node[data] (4) {\footnotesize \texttt{4}};
\path (add1)+(0.375,-0.8) node[op,circle,inner sep=0.05cm] (mult) {\scriptsize \textsc{mult}};
\path (add1)+(0.75,0) node[data] (2) {\footnotesize \texttt{2}};

\path[->] (r3) edge[] node[] {} (add1);
\path[->] (4) edge[] node[] {} (add1);
\path[->] (add1) edge[] node[] {} (mult);
\path[->] (2) edge[] node[] {} (mult);

\path (const)+(1,0) node [] (arrow1) {};
\path (arrow1.center)+(1,0) node[] (arrow2) {};
\path[->,line width=3pt] (arrow1.center) edge [] node[] {} (arrow2.center);

\path (arrow2)+(0.8,0.8) node[data] (r3p) {\footnotesize R3};
\path (r3p)+(0.375, -0.8) node[op,circle,inner sep=0.05cm] (multp) {\scriptsize \textsc{mult}};
\path (r3p)+(0.75,0) node[data] (2p) {\footnotesize \texttt{2}};
\path (multp)+(0.375,-0.8) node[op,circle,inner sep=0.1cm] (add2) {+};
\path (multp)+(0.75,0) node[data] (8) {\footnotesize \texttt{8}};

\path[->] (r3p) edge[] node[] {} (multp);
\path[->] (2p) edge[] node[] {} (multp);
\path[->] (multp) edge[] node[] {} (add2);
\path[->] (8) edge[] node[] {} (add2);

\end{tikzpicture}
}
\end{center}

\section{Sample signature definition}\label{appendix:signature}
Given below is a snippet taken from the (N)LFSR signature bundled with our implementation of the framework.
\begin{verbatim}
IDENTIFIER (Non-)Linear feedback shift register

VARIANT A
...

VARIANT C
TRANSIENT layer0:OR(AND(1,OPAQUE),OPAQUE<<1);
TRANSIENT layer1:OR(AND(1,OPAQUE),layer0<<1);
TRANSIENT layer2:OR(AND(1,OPAQUE),layer1<<1);
layer3:OR(AND(1,OPAQUE),layer2<<1);
\end{verbatim}
An (N)LFSR can be implemented in software by various means. For e.g., rather than shifting to the left, the register may shift to the right instead, placing the new bit generated by the feedback function at the most significant position. Furthermore, a left shift of one bit is equivalent to a multiplication with $2$, or an addition with itself. Also, the newly generated bit is normally appended to the register through a bitwise or. However, directly after a shift operation is performed, the vacant bit is always $0$. Hence, using an exclusive-or, or even an addition instead is equivalent. Due to these naturally occurring variations, several variants of the signature are defined. In this example, we take a closer look at variant \texttt{C}, which is the most typical.

As discussed in Section~\ref{sec:evaluation}, we take $n=4$. Hence, the signature should capture $4$ iterations of an (N)LFSR. Each iteration, the register shifts one position to the left, and a new bit is generated by an unknown feedback function $L$ and placed at position $0$ by means of a bitwise or. Each round refers to the previous through its label, i.e. \texttt{layer[0-3]}. The initial state is the result of an unknown initialization function, hence represented by \texttt{OPAQUE}. $L$ is also unknown, and thus represented by \texttt{OPAQUE}. However, it is known to produce a single output bit. Therefore, it can be assumed that the single bit is obtained through a bitwise-and with $1$, before being inserted into the register by means of a bitwise or. Finally, all iterations except the last form intermediate steps towards the register's final value. By specifying the \textsc{transient} keyword, we allow the broker to translate the intermediate steps into a more optimized DFG representation.

\section{Implementation}

An implementation of the framework described in this paper is available for download\footnote{\url{https://github.com/wheres-crypto/wheres-crypto}}. It comes in the form of a plug-in for the popular IDA disassembler. At the time of writing, support is implemented for 32 bits ARM binaries. The architecture is modular, and expanding support to other architectures is relatively straightforward. Figure~\ref{fig:poc} shows a sample analysis report, and a DFG plot generated by our implementation.

\end{document}